%% file: main-arxiv.tex
\documentclass[acmsmall,nonacm]{acmart}

\usepackage{booktabs} 
\usepackage[ruled,noend]{algorithm2e} 
\usepackage{thmtools}
\usepackage{thm-restate}

\SetAlFnt{\small}
\SetAlCapFnt{\small}
\SetAlCapNameFnt{\small}
\SetAlCapHSkip{0pt}
\IncMargin{-\parindent}

\usepackage{tikz}
\usetikzlibrary{fit,backgrounds}
\usepackage{xcolor}
\usepackage{tcolorbox}
\usepackage[normalem]{ulem}

\include{header}

\setcitestyle{authoryear}

\title{Position Auctions in AI-Generated Content}

\author{Santiago Balseiro}
\affiliation{%
  \institution{Columbia University \& Google Research}
}
\email{srb2155@columbia.edu}
\author{Kshipra Bhawalkar}
\author{Yuan Deng}
\author{Zhe Feng}
\author{Jieming Mao}
\author{Aranyak Mehta}
\author{Vahab Mirrokni}
\author{Renato Paes Leme}
\author{Di Wang}
\author{Song Zuo}
\affiliation{%
  \institution{Google Research}
}
\email{{kshipra,dengyuan,zhef,maojm,aranyak,mirrokni,renatoppl,wadi,szuo}@google.com}

\begin{abstract}
  \input{abstract}
\end{abstract}

\begin{document}

\begin{titlepage}

\maketitle

\vspace{1cm}
\setcounter{tocdepth}{2} 
\tableofcontents

\end{titlepage}

\input{main-body.tex}
\bibliographystyle{ACM-Reference-Format}
\bibliography{references}
\newpage
\appendix
\input{appendix}
\end{document}

%% file: header.tex
\newcommand{\X}{\mathcal{X}}

\newcommand{\R}{\mathbb{R}}

\newcommand{\x}{\boldsymbol{x}}
\newcommand{\bchi}{\boldsymbol{\chi}}
\renewcommand{\v}{\boldsymbol{v}}
\renewcommand{\b}{\boldsymbol{b}}
\newcommand{\e}{\boldsymbol{e}}
\newcommand{\q}{\boldsymbol{q}}
\renewcommand{\t}{\boldsymbol{t}}
\newcommand{\p}{\boldsymbol{p}}

\newcommand{\pir}{\pi^{\text{restricted}}}
\newcommand{\welr}{\operatorname{Welfare}^{\text{restricted}}}
\newcommand{\welbase}{\operatorname{Welfare}^{\text{base}}}
\newcommand{\pibase}{\pi^{\text{base}}}
\newcommand{\pib}{\pi^{\text{budgeted}}}
\newcommand{\welb}{\operatorname{Welfare}^{\text{budgeted}}}
\newcommand{\wel}{\operatorname{Welfare}}
\newcommand{\zs}{\operatorname{ZS}}

\newcommand{\pibka}{\pi^{\text{budgeted}-k,\alpha}}
\newcommand{\welbka}{\operatorname{Welfare}^{\text{budgeted}-k,\alpha}}
\newcommand{\pka}{p^{k,\alpha}}
\newcommand{\xka}{\x^{k,\alpha}}

\newtheorem{theorem}{Theorem}[section]
\newtheorem{lemma}[theorem]{Lemma}

\newtheorem{corollary}[theorem]{Corollary}

\newtheorem{definition}[theorem]{Definition}



%% file: abstract.tex
We consider an extension to the classic position auctions in which  sponsored creatives can be added within AI generated content rather than shown in predefined slots. New challenges arise from the natural requirement that sponsored creatives should smoothly fit into the context. With the help of advanced LLM technologies, it becomes viable to accurately estimate the benefits of adding each individual sponsored creatives into each potential positions within the AI generated content by properly taking the context into account. Therefore, we assume one click-through rate estimation for each position-creative pair, rather than one uniform estimation for each sponsored creative across all positions in classic settings. As a result, the underlying optimization becomes a general matching problem, thus the substitution effects should be treated more carefully compared to standard position auction settings, where the slots are independent with each other.

In this work, we formalize a concrete mathematical model of the extended position auction problem and study the welfare-maximization and revenue-maximization mechanism design problem. Formally, we consider two different user behavior models and solve the mechanism design problems therein respectively. For the Multinomial Logit (MNL) model, which is order-insensitive, we can efficiently implement the optimal mechanisms. For the cascade model, which is order-sensitive, we provide approximately optimal solutions.

%% file: main-body.tex
\section{Introduction}

Position auctions \citep{varian2007position}, the foundational mechanism for sponsored search, have facilitated trillions of dollars in business over the past decades. These auctions sell positions that are identical except for their rank. Under the commonly used separability assumption,\footnote{The separability assumption says that the click-through rate $p_{ij}$ of ad $i$ in position $j$ can be decomposed into an ad-specific factor $\alpha_i$ and a position specific factor $\beta_j$, i.e., $p_{ij} = \alpha_i \cdot \beta_j$ for all ad $i$ and position $j$.} the competition for each position becomes effectively identical, leading to a uniform preference order among buyers and a significantly simplified optimization problem.

Inspired by advances in generative AI, we consider an extended position auction setting in which buyers compete for the placement of sponsored creatives within AI-generated commercial content. For example, Microsoft is testing embedding ads in the new Bing generative AI search engine~\citep{bing2023} and Perplexity started to show ads inside their generative AI chatbot~\citep{perplexity2024}. See Figure~\ref{fig:ads-in-genai} for an example illustration.
\begin{figure}[h]
\centering
\vspace{-0.25cm}
\begin{tcolorbox}[colback=white,arc=0pt,outer arc=0pt,boxrule=0.5pt,toprule=0.5pt,bottomrule=0.5pt,rightrule=0.5pt,leftrule=0.5pt,width=0.99\textwidth]
{Healthy pets need proper care: nutritious food, exercise, and vet visits. Understanding your pet's specific needs is important.}
\colorbox{blue!20}{\textit{\textcolor{blue}{\uline{[Sponsor A]}} offers resources and pet food options.}}
{Choose quality food, avoiding artificial additives. Regular vet check-ups are also key.}
\colorbox{blue!20}{\textit{\textcolor{blue}{\uline{[Sponsor B]}}'s grain-free dog food is one option.}} ...
\end{tcolorbox}
\vspace{-0.25cm}
\caption{An Illustration of Showing Ads within AI-Generated Contents.}
\label{fig:ads-in-genai}
\end{figure}

This novel application introduces context as a significant factor influencing click-through rates. Consequently, buyer preference orders for positions are no longer uniform, and hence the separability assumption no longer holds, requiring click-through rate measurement for each position-creative pair. Fortunately, the advanced semantic understanding of large language models makes this feasible.


However, two significant challenges arise when the separability assumption is removed. The first is that the underlying optimization is no longer a ranking problem but a general matching problem. The second, a direct consequence of the first, is the need for an explicit model of substitution effects, capturing user behavior in the presence of multiple sponsored creatives.

In this study, we address the design of optimal welfare-maximization and revenue-maximization position auctions within the context of two user behavior models: (1) the canonical {\em multinomial logit (MNL) model} \citep{luce1959individual,mcfadden_conditional_1974,plackett1975analysis}, where the user is insensitive to the presentation order of creatives; and (2) the {\em cascade model} \citep{craswell2008experimental}, where order matters. In particular, computational tractability is a primary concern, which ensures the practical viability of the proposed auction mechanisms. Therefore, we focus on the design of computationally efficient optimization algorithms, which can be used to obtain optimal welfare-maximization and revenue-maximization position auctions via Vickrey-Clarke-Groves (VCG) mechanisms \citep{vickrey1961counterspeculation,clarke1971multipart} and Myerson auctions \citep{myerson1981optimal}.

\subsection{Our Contribution}

\paragraph{\bf Model Formalization}
In Section \ref{sec:model}, we formalize the model of the extended position auctions, motivated by the novel application of placing sponsored creatives within AI-generated commercial content. 
As we drop the separability assumption, the input of the problem naturally includes a matrix of click-through rates and a vector of bids. For each position-creative pair, the click-through rate represents the probability that the creative being clicked when shown in the corresponding position while no other creative is shown, which we call the {\em standalone click-through rate}.

In the case where multiple creatives are shown, the user's attention is distributed over them, and hence the substitution effect emerges. To quantify such effects, certain {\em user behavior models} must be introduced, and in particular, we study the multinomial logit (MNL) model and the cascade model, respectively. Given a user behavior model, the final click-through rates for each creative is then a function of both the allocation of creatives and their standalone click-through rates in the allocated positions.

\paragraph{\bf Winner Determination Problem}
In Section \ref{sec:wdp}, we formalize the {\em winner determination problem} (WDP), prove the monotonicity of its solution, and demonstrate the reduction from the welfare-maximization and the revenue-maximization problem to it.

The WDP is to find the ``winning'' allocation out of the feasible set of allocations that maximizes the total payoff assuming winning advertisers pay their bids for any given ``bid'' vector. With an oracle that can solve the WDP exactly, one can easily implement the VCG mechanism: the allocation is the solution of the WDP with advertiser values as the ``bids'', and the VCG payment can be computed from the optimal payoff of the leave-one-out WDPs.

For the revenue-maximization problem, solving the WDP with advertiser {\em virtual values} as the ``bids'' yields the optimal allocation, and the corresponding payment can be computed following the envelope theorem \citep{milgrom2002envelope}. According to the envelope theorem, one must prove that the allocation function is monotonic in the sense that each advertiser's click-through rate (weakly) increases as its bid increases (while other bids fixed), and then the corresponding payment can be computed based on the integral of the allocation function. Fortunately, the monotonicity is automatically guaranteed as long as the WDP is solved exactly (Lemma \ref{lem:wdp-monotone}).

A useful side note is that the envelope theorem works even when the WDP is solved approximately, except that the monotonicity must to be proved separately for the implied allocation function. Thus, when an approximate solver for the WDP is given, one can invoke the envelope theorem to construct approximately welfare-maximization and revenue-maximization mechanism as long as the allocation function is monotonic.




\paragraph{\bf Optimal Auctions in MNL Model}
In Section \ref{sec:mnl}, we present both the welfare-maximization mechanism and the revenue-maximization mechanism under the MNL model (Theorem \ref{thm:mnl}).
Given the reduction to the WDP described above, it remains to provide a computationally efficient algorithm that can solve the WDP exactly. We achieve this by first relax the WDP as a linear program, and then construct an optimal integral solution to the WDP from any solution of the relaxed linear program (Lemma \ref{lem:mnl-int}). 
We emphasize that the WDP itself is a mixed integer program with a non-linear objective. Simple relaxations do not make it a linear program. We instead develop an innovative transformation from using allocations as the decision variables to two carefully chosen intermediate quantities as the decision variables. The choice of the intermediates not only admits a linear program formulation, but also guarantees a mapping back to integral allocation variables without any loss in the objective value.



\paragraph{\bf Approximately Optimal Auctions in Cascade Model}
In Section \ref{sec:cascade}, we present the approximately optimal welfare-maximization mechanism and the revenue-maximization mechanism under the Cascade model (Corollary \ref{cor:cascade-mech}).
Similar to the mechanism design under the MNL model, the key is the efficient algorithm that solves the WDP approximately. To achieve this, we introduce an intermediate problem that 4-approximates the WDP and then design algorithms to solve the intermediate problem (see Section \ref{sec:4apx}). Interestingly, the intermediate problem is similar to the so-called {\em budgeted matching problem} \citep{10.5555/1788814.1788839}, and we present a PTAS that solves it up to any $\epsilon > 0$ gap. However, the resulting allocation function is not necessarily monotonic. As an alternative, we propose an efficient algorithm with an approximation ratio of $O(\ln m)$, where $m$ is the number of positions, and prove the monotonicity of the corresponding allocation rule. We leave it as an interesting open problem to close the gap between the WDP and the WDP with a monotone allocation constraint, particularly, whether it is possible to obtain a monotonic allocation with a constant approximation ratio.




\subsection{Related Work}

\paragraph{Position Auctions and Online Advertising} 
In the seminal work, \citet{varian2007position} first introduced the game-theoretic model of position auctions for sponsored search. \citet{edelman2007internet} provided a comprehensive analysis of the Generalized Second-Price auction, examining its efficiency and the strategic behavior of advertisers. \citet{blumrosen2008position} introduced a setting with non-uniform conversion rates, where the conversion probability given a click varies depending on the ad's position. \citet{athey2011position} further advanced the field by incorporating consumer search behavior into the model, recognizing that consumer choices and search strategies influence the value of different ad positions. \citet{thompson2009computational} developed ``computational mechanism analysis'' to quantitatively compare different position auction formats, providing insights into their efficiency and effectiveness. More recently, \citet{gravin2024bidder} addressed the bidder selection problem, proposing efficient algorithms for selecting a subset of advertisers to participate in the auction.


\paragraph{Cascade Model}
The cascade model \citep{craswell2008experimental}, in which users scan slots sequentially from top to bottom, has been shown empirically to explain user choices effectively. Several studies have explored the cascade model and its extensions (see, e.g.,~\citep{KempeM08,AggarwalFMP08,GomesIM09,ChuNZ20}).~\citet{KempeM08} and~\citet{AggarwalFMP08} demonstrate that the winner determination problem for the cascade model can be solved in polynomial time using dynamic programming. Interestingly,~\citet{KempeM08} also examines a more general model with position multipliers, presenting a $4$-approximation algorithm and a QPTAS, but leaves open the question of whether the problem is NP-hard.

In contrast, non-sequential choice models propose that consumers evaluate all advertisements before making a selection.~\citet{JeziorskiSegal15} provides empirical evidence supporting that non-sequential models being better at explaining consumer choices.~\citet{GhoshM08} investigates a non-sequential model with a public choice model and private advertiser values, showing that an efficient allocation can be implemented with a VCG mechanism. While the winner determination problem can be NP-hard in general,~\citet{GhoshM08} offers a polynomial time algorithm using dynamic programming when preferences are single-peaked.
In the previous results, the choice probabilities are public, and the only private information is the advertisers' values for being chosen.~\citet{athey2011position} gives a notable exception, studying the bidding equilibrium under a cascade model where consumers need to click on ads to discover their quality.

\paragraph{LLM motivated topics}
Inspired by the successful application driven by large language models, many recent works start looking into economic problems that intersect with generative AI. \citet{duetting2024mechanism,dubey2024auctions,soumalias2024truthful,hajiaghayi2024ad} consider mechanism design in the presence of LLMs as an extension for online advertising with generative AI. \citet{sun2024mechanism} study the mechanism design problem with strategic agents during the LLM fine-tuning stage. \citet{lu2024eliciting,wu2024elicitationgpt} leverage the sematic understanding ability of LLMs to expand informative elicitation mechanisms to natural languages. Besides the above mechanism design problems, there are more works focusing on understanding the economic behaviors and strategical abilities of LLMs, e.g., \citep{chen2023put,deng2024llms,lore2023strategic,fan2024can,raman2024rationality,meta2022human,mao2023alympics,gemp2024states}.

\section{Model}\label{sec:model}

We introduce a model of extended position auctions with $m$ positions and $n$ advertisers, where the rendering order of positions is chosen by the auctioneer. In the application of position auctions within AI-generated commercial content, a position consists of the nearby context and a potential slot for sponsored creatives. In particular, the auctioneer can choose an order of the contents (and hence the positions) and instruct LLMs to generate the final content aligned with the selected order. 

We denote by $\X$ the set of feasible {\em allocations} of advertisers to position, and $\Sigma$ the set of feasible rendering orders of positions or simply {\em permutations}. We call the combination of an allocation and a permutation an {\em augmented allocation}. Allocations are feasible if each advertiser is allocated at at most one position, each position is allocated at most one advertiser, and the total number of allocated positions is at most $K \le m$. The set of feasible allocations is thus given by
\[
\X = \left\{ \x \in \{0,1\}^{n\times m} : \sum_{j=1}^m x_{ij} \le 1\,, \forall i \in [n]\,, \sum_{i=1}^n x_{ij} \le 1\,, \forall j \in [m]\,, \sum_{i=1}^n \sum_{j=1}^m x_{ij} \le K \right\}\,.
\]
The set of feasible permutations $\Sigma$ is simply the set of all permutations over $[m]$.

Advertisers' private information is their value for a click. We denote by $v_i \ge 0$ the value of advertiser $i \in [n]$. We assume that values are independently distributed with an absolutely continuous distribution. We denote by $F_i$ the cumulative distribution and by $f_i$ the density of advertiser $i$'s value, respectively.

Our model allows for externalities between advertisers' click-through rates. We denote by $\pi_i : \X \times \Sigma \rightarrow [0,1]$ a function that maps the augmented allocation $\bchi := (\x, \sigma) \in \X\times \Sigma$ to a click-through rate $\pi_i(\bchi)$ or equivalently $\pi_i(\x, \sigma)$. Click-through rates are usually estimated by platforms using sophisticated predictions models and we assume that $\pi_i(\x, \sigma)$ is common knowledge.

By the revelation principle, we can restrict attention to direct mechanisms in which every advertiser reports their valuation. A mechanism is a pair $(\q,\t)$ characterized by an augmented allocation function $\q: \R_+^n \rightarrow \Delta(\X \times \Sigma)$ and payment functions $\t: \R_+^n \rightarrow \R_+^n$. The augmented allocation function maps a vector of reported values $\v \in \R_+^n$ to a lottery over augmented allocations $\q(\v) \in \X \times \Sigma$. We denote by $q(\v)[\bchi]$ or equivalently $q(\v)[\x, \sigma]$ the probability of choosing augmented allocation $(\x, \sigma) \in \X \times \Sigma$. The payment function gives the expected payment $t_i(\v)$ made by advertiser $i$ when reported values are $\v \in \R_+^n$.

For simplicity, we may abuse the above notations by omitting $\sigma$ (e.g., $\pi_i(\x, \sigma)$ as $\pi_i(\x)$, $q(\v)[\x, \sigma]$ as $q(\v)[\x]$, $\X \times \Sigma$ as $\X$) when the value of $\sigma$ is clear from the context or the outcomes are insensitive to $\sigma$, e.g., $\pi_i(\x, \sigma) \equiv \pi_i(\x),~\forall \sigma \in \Sigma$.

Advertisers are quasi-linear utility maximizers, i.e., they maximize the difference between their expected value and the payments made to the mechanism. The expected utility of advertiser $i$ when their value is $v_i \in \R_+$ and the vector of reported values of $\hat \v \in \R_+^n$ is
\[
    u_i(\hat \v, v_i) = v_i \sum_{\x \in \X, \sigma \in \Sigma} q(\hat \v)[\x, \sigma] \pi_i(\x, \sigma)- t_i(\hat \v)\,.
\]
We restrict attention to mechanisms that are dominant-strategy incentive compatible (IC) and individually rational (IR). A mechanism is IC if every buyer is better off reporting their value truthfully regardless of the reported values of their competitors. That is,
\begin{align}\label{eq:IC}
    u_i((v_i,\hat \v_{-i}), v_i) \ge u_i(\hat \v, v_i)\,,\quad \forall i \in [n]\,, \hat \v \in \R_+^n\,, v \in \R_+\,,\tag{IC}
\end{align}
where we denote by $\hat \v_{-i} \in \R_+^{n-1}$ the vector of advertisers values excluding advertiser $i$. A mechanism is IR if advertisers get a non-negative utility from participating
\begin{align}\label{eq:IR}
    u_i((v_i,\hat \v_{-i}), v_i) \ge 0 \,,\quad \forall i \in [n]\,, \hat \v \in \R_+^n\,, v \in \R_+\,.\tag{IR}
\end{align}

\subsection{User Behavior Models}

We consider different models for click-through rate that incorporate user behavior. The models take as input a matrix $\p \in \R^{n\times m}$, where $p_{ij}$ is the probability that a user clicks on advertisement $i$ when displayed in position $j$ when no other advertisement is shown in any other position. We consider two popular user behavior model that incorporate substitution effects between advertisers: a multinomial logit model and a cascade model.

\paragraph{Multinomial Logit Model (MNL)} The MNL model \citep{luce1959individual,mcfadden_conditional_1974,plackett1975analysis} assumes that the user obtains independent random utilities from each ad, which are drawn from independent Gumbell distributions. The users considers all advertisements simultaneously and chooses one with the highest utility. MNL models are used extensively in marketing, economics, and operations literature. These models are popular because they are interpretable, flexible, easy to estimate, and extendable. Let $\rho_{ij} = \operatorname{logit}(p_{ij}) = \log(p_{ij}/(1-p_{ij}))$ be the log-odds of showing advertiser $i$ in slot $j$. The click-through rate of advertiser $i$ under an augmented allocation $(\x, \sigma)$ is 
\[
    \pi_i(\x, \sigma) \equiv \pi_i(\x) = \frac{\sum_{j \in [m]} x_{ij} \exp(\rho_{ij})}{1 + \sum_{i' \in [n]} \sum_{j \in [m]} x_{i'j} \exp(\rho_{i'j})}\,.
\]
Since the click-through rate here is insensitive to the permutation $\sigma$, we will omit $\sigma$ from now on whenever the MNL model is adopted. When advertiser $i$ is shown in slot $j$ and no ads are shown in other slots, we obtain that $\pi(\e^{ij}) = p_{ij}$ where $\e^{ij} \in \X$ is the allocation with one in position $(i,j)$ and zero otherwise.


\paragraph{Cascade Model} The Cascade model \citep{craswell2008experimental} assumes that the user reads through the ads in the given rendering order and leaves the system once a click happens.
Assuming all ad-clicking decisions are independent, the click-through rate of advertiser $i$ under allocation $x\in\X$ and permutation $\sigma$ is discounted by the product of $(1 - p_{i'j'})$ for all ads $i'$ shown before it, i.e.,

\[
    \pi_i(\x, \sigma) = \sum_{j \in [m]} x_{ij} p_{ij} \prod_{j': \sigma(j')<\sigma(j)} \left( 1 - \sum_{i' \in [n]}  x_{i'j'} p_{i'j'} \right)\,.
\]

\subsection{Objectives}

We consider two objectives for the ad platform, welfare and revenue maximization. When the mechanism is IC and IR, the advertisers report truthfully and we can evaluate the metrics by taking expectation over the distribution of advertisers' values.

The welfare is the sum of the platform revenue and the advertiser utility. Since payments are internal transfers of wealth, they cancel out and the expected welfare of a mechanism $(\q,\t)$ is
\[
    \operatorname{Welfare}(\q,\t) = \mathbb E_{\v}\left[ \sum_{i\in[n]}v_i \sum_{\x \in \X, \sigma \in \Sigma} q(\v)[\x, \sigma] \pi_i(\x, \sigma) \right]\,.
\]
The revenue of the ad platform is just the sum of the advertisers' payments:
\[
    \operatorname{Revenue}(\q,\t) = \mathbb E_{\v}\left[ \sum_{i\in[n]} t_i(\v) \right]\,.
\]

\section{Winner Determination Problem and Reduction from Mechanism Design Problems}\label{sec:wdp}

In this section, we first formally define the {\em winner determination problem} (WDP) and then show how the optimal welfare-maximization auction and the optimal revenue-maximization auction are constructed with an exact solver of the WDP. We also extend the reduction to the case where only an approximate solver is given, in which the corresponding auctions downgrade to approximately optimal auctions.

\subsection{Winner Determination Problem}

The WDP determines, given a vector of ``bids'' $\b \in \R^n$, the augmented allocation $\bchi = (\x, \sigma) \in \X\times\Sigma$ that maximizes the payoff assuming advertisers pay their bids. This is given by
\[
    \bchi^*(\b) = (\x^*(\b), \sigma^*(\b)) \in \arg\max_{\bchi \in \X\times\Sigma} \sum_{i\in[n]} b_i \pi_i(\bchi).
\]
In following sections, we discuss how to efficiently solve the WDP. Before that, we prove monotonicity of click-through rates of an advertiser in terms of its bid when the WDP is solved exactly.

\begin{lemma}\label{lem:wdp-monotone}
    Fix an advertiser $i \in [n]$. Let $\bchi^*(\b)$ be an optimal augmented allocation for a bid vector $\b$ and $\bchi^*(\hat \b)$ be an optimal augmented allocation for the vector $\hat \b$ with $\hat b_i \ge b_i$ and $\hat b_{i'} = \hat b_{i'}$ for $i'\neq i$. Then, $\pi_i( \bchi^*(\hat \b) ) \ge \pi_i( \bchi^*(\b))$.
\end{lemma}

\begin{proof}
Let $\boldsymbol{\pi} = \big( \pi_i( \bchi^*(\b)) \big)_{i \in [n]}$ and $\hat{\boldsymbol{\pi}} = \big( \pi_i( \bchi^*(\hat \b)) \big)_{i \in [n]}$. We prove the result by contradiction. Suppose $\hat b_i > b_i$ but $\hat \pi_i < \pi_i$. By the optimality of $\bchi^*(\hat \b)$, we have that
\begin{equation}\label{eq:mono1}
    \hat b_i \hat \pi_i + \sum_{i'\neq i} b_i \hat \pi_i \ge \hat b_i \pi_i + \sum_{i'\neq i} b_i  \pi_i\,.
\end{equation}
Because $\hat b_i > b_i$ and $\hat \pi_i < \pi_i$, we have $(\hat b_i - b_i) \cdot (\hat \pi_i - \pi_i) < 0$, which implies after re-arranging that
\begin{equation}\label{eq:mono2}
    b_i \hat \pi_i + \hat b_i \pi_i > \hat b_i \hat \pi_i + b_i \pi_i\,.
\end{equation}
Summing \eqref{eq:mono1} and \eqref{eq:mono2} we obtain that
\[
    b_i \hat \pi_i + \sum_{i'\neq i} b_i \hat \pi_i > b_i \pi_i + \sum_{i'\neq i} b_i  \pi_i\,,
\]
which contradicts the optimality of $\bchi^*(\b)$.
\end{proof}

\subsection{Mechanism Design Problem}\label{subsec:mechanism}
In this subsection, we show how the welfare-maximization auction and the revenue-maximization auction can be implemented when an exact optimization algorithm of the WDP is given.

\paragraph{Welfare Maximization} The ad platform's problem is to solve
\[
    \max_{(\q,\t) : \eqref{eq:IC} \text{ and } \eqref{eq:IR}} \operatorname{Welfare}(\q,\t)\,.
\]
If we ignore the individual rationality (IR) and incentive compatibility (IC) constraints, we can write the problem as 
\[
    \max_{(\q,\t)} \operatorname{Welfare}(\q,\t) = \max_{\q} \mathbb E_{\v}\left[ \sum_{i\in[n]}v_i \sum_{\bchi \in \X \times \Sigma} q(\v)[\bchi] \pi_i(\bchi) \right] = \mathbb E_{\v}\left[ \max_{\bchi \in \X \times \Sigma} \sum_{i\in[n]}v_i \pi_i(\bchi) \right]\,,
\]
where the first equation follows because payments do not enter the objective, and the second because the platform should optimize pointwise over values and pick the allocation with the highest values. Denote the allocation that maximizes welfare for each realization of advertisers' values by $\bchi^*(\v)$. 

We can implement a welfare-maximizing augmented allocation $\bchi^*(\v)$ truthfully using the Vickrey–Clarke–Groves (VCG) mechanism. Under the VCG mechanism, the platform asks advertisers to report their values $\v$, and then it chooses the augmented allocation $\bchi^*(\v)$, and charges each advertiser the externality they impose on others. That is, the augmented allocation is $q(\v)[\bchi] = 1$ for $\bchi = \bchi^*(\v)$ and $q(\v)[\bchi] = 0$ otherwise and the payment of advertiser $i$ is
\begin{align}\label{eq:payments-vcg}
    t_i(\v) = \sum_{i'\neq i} v_{i'} \pi_{i'}(\bchi^*(0,\v_{-i})) - \sum_{i'\neq i} v_{i'} \pi_{i'}(\bchi^*(\v))\,.   
\end{align}
The VCG mechanism is IC. Because values are non-negative, we have that payments satisfy $t_i(\v) \ge 0$ and the mechanism is also IR.

\paragraph{Revenue Maximization} The ad platform problem is to solve
\[
    \max_{(\q,\t) : \eqref{eq:IC} \text{ and } \eqref{eq:IR}} \operatorname{Revenue}(\q,\t)\,.
\]
We can use the Envelope theorem \citep{milgrom2002envelope} to characterize optimal mechanisms. Let 
\[
    y_i(\v) = \sum_{\bchi \in \X \times \sigma} q( \v)[\bchi] \pi_i(\bchi)\,,
\]
be the probability that advertiser $i$ is allocated (in any slot) under a mechanism 
$(\q,\t)$. We refer to $y_i(\v)$ as the \emph{total click probability} of advertiser $i$. Then, a mechanism is IC if and only if the total click probability $y_i(\v)$ is non-decreasing in $v_i$ for all advertiser $i \in [n]$ and payments satisfy
\begin{align}\label{eq:envelope}
    t_i(\v) = t_i(0,\v_{-i}) + v_i y_i(\v) - \int_0^{v_i} y_i(z,\v_{-i}) \text{d} z\,. 
\end{align}
Moreover, a mechanism is IR if $t_i(0,\v_{-i}) \le 0$. At optimality, we should set the payment of the lowest type to zero, i.e., $t_i(0,\v_{-i})=0$. Using integration by parts, we can write the revenue of an optimal mechanism as 
\[
    \operatorname{Revenue}(\q,\t) = \mathbb E_{\v}\left[ \sum_{i\in[n]} \left(v_i - \frac {1-F_i(v_i)} {f_i(v_i)} \right) y_i(\v) \right] = \mathbb E_{\v}\left[ \sum_{i\in[n]} \phi_i(v_i) \sum_{\bchi \in \X \times \Sigma} q(\v)[\bchi] \pi_i(\bchi) \right] \,,
\]
where the last equation follows from using our formula for the total click probability $y_i(\v)$ and letting $\phi_i(v_i) = v_i - (1-F_i(v_i))/f_i(v_i)$ be the virtual value of advertiser $i \in [n]$.

A similar argument to the welfare case implies that the platform should choose the allocation that maximizes virtual values, i.e., choose $\bchi^*(\b)$ with $b_i = \phi_i(v_i)$. This allocation can be implemented truthfully if payments are chosen to satisfy \eqref{eq:envelope} and the total click probability is non-decreasing. A sufficient condition for the total click probability $y_i(\v)$ to be monotonic is that virtual values $\phi_i(v_i)$ are non-decreasing and that $\pi_i(\bchi^*(\v))$ is non-decreasing in $v_i$ for every advertiser $i$ and valuations $v_{-i}$ of competitors.

\subsection{Mechanism Design with Approximate Solution of WDP}\label{subsec:approx-mech}

A useful observation here is that the envelope theorem works even when the WDP is solved approximately. However, additional care is needed for proving the monotonicity of the induced allocation function, as Lemma \ref{lem:wdp-monotone} does not apply to approximate solvers.

Following the reductions in Section \ref{subsec:mechanism}, an approximately optimal solution to the WDP yields approximately optimal welfare (or revenue) from the mechanism, and the corresponding payment is computed according to the envelope theorem.

\section{Multinomial Logit Model}\label{sec:mnl}

A special case of the MNL model is when the log-odds $\rho_{ij}$ are additively separable, i.e., $\rho_{ij} = \rho_i + \rho_j$. \citet{abeliuk2014optimizing} showed that the winner determination problem can be solved in quadratic time when log-odds are additively separable. Additive separability allows no interaction between positions and advertisers, as every advertiser equally prefers each position. We next discuss how to solve the WDP for general log-odds efficiently. Since the MNL model is constant to the permutation, we omit the permutation $\sigma$ and focus on the allocation $x$ for the entire section.

The key of the results in this section is an efficient algorithm solving the WDP for general log-odds. 
Taking the algorithm as given, with the reduction introduced in Section \ref{subsec:mechanism}, we can implement both the welfare-maximization auction and the revenue-maximization auction. Formally we have the following theorem:

\begin{theorem}[Optimal Auctions in MNL model]\label{thm:mnl}
  Under the MNL model, there exists:
  \begin{itemize}
      \item A welfare-maximization auction that is computationally efficient and satisfies IC and IR;
      \item A revenue-maximization auction that is computationally efficient and satisfies $\epsilon$-IC and IR.
  \end{itemize}
\end{theorem}
The $\epsilon$ loss on IC for the revenue-maximization auction comes from the $\epsilon$-accuracy loss of payment computation, where one needs to compute the integral of the allocation function with potentially exponentially many pieces. 
The $\epsilon$ loss then comes from discretizing the space of the integral to guarantee computational efficiency.
We now proceed to discuss how to efficiently solve the WDP under the MNL model, despite of the non-linear objective:
\[
    \max_{\x \in \X}  \frac{\sum_{i \in [n]} \sum_{j \in [m]} b_i x_{ij} \exp(\rho_{ij})}{1 + \sum_{i' \in [n]} \sum_{j \in [m]} x_{i'j} \exp(\rho_{i'j})}\,.
\]
Consider the convex relaxation in which we optimize over the convex hull of $\X$ given by
\[
    \bar \X = \left\{ \x \in \mathbb R_+^{n\times m} : \sum_{j=1}^m x_{ij} \le 1\,, \forall i \in [n]\,, \sum_{i=1}^n x_{ij} \le 1\,, \forall j \in [m]\,, \sum_{i=1}^n \sum_{j=1}^m x_{ij} \le K \right\}\,.
\]
We present a transformation that allows us to write the linear-fractional problem as a linear program. We introduce the following variables
\[
    y_{ij} = \frac{x_{ij}} {1 + \sum_{i' \in [n]} \sum_{j \in [m]} x_{i'j} \exp(\rho_{i'j})}
\qquad \text{and} \qquad
    z = \frac 1 {1 + \sum_{i' \in [n]} \sum_{j \in [m]} x_{i'j} \exp(\rho_{i'j})}\,.
\]

We can equivalently solve the following linear program:
\begin{equation}\label{eq:lp}
\begin{aligned}
\max_{y \in \mathbb R_+^{n \times m}, z \in \mathbb R}  \, & \sum_{i \in [n]} \sum_{j \in [m]} b_i y_{ij} \exp(\rho_{ij}) \\
    \text{s.t.}\,
    & \sum_{j=1}^m y_{ij} \le z\,, \forall i \in [n]\,,\\ & \sum_{i=1}^n y_{ij} \le z\,, \forall j \in [m]\,, \\
    & \sum_{i=1}^n \sum_{j=1}^m y_{ij} \le K z\,,\\
    & \sum_{i \in [n]} \sum_{j \in [m]} y_{ij} \exp(\rho_{ij}) + z = 1\,.\\
    \end{aligned}
\end{equation}

\begin{lemma}\label{lem:mnl-int} Suppose $\b \neq 0$. Let $(\boldsymbol y,z)$ be an optimal solution to \eqref{eq:lp}. Then, an optimal integral matching can be obtained by setting
\[
    \x = \frac {\boldsymbol{y}} z\,.
\]
\end{lemma}
\begin{proof}
First, we argue that the WDP problem and its convex relaxation have the same optimal solution, i.e., the optimal solution of the relaxation is integral. Note that the objective $\sum_{i \in [n]} b_i \pi_i(\x)$ is convex in $\x$ because linear-fractional functions are convex on its domain, see Chapter 2.3.3 of \cite{boyd2004convex}. Because we are maximizing a convex and continuous function over a compact and convex set, we obtain from Bauer maximum principle that the problem attains its maximum at some extreme point of that set. Because the extreme points of the matching polytope $\bar \X$ are integral matchings, we obtain that
\[
    \arg\max_{\x \in \X} \sum_{i \in [n]} b_i \pi_i(\x) = \arg\max_{\x \in \bar \X} \sum_{i \in [n]} b_i \pi_i(\x)\,,
\]
that is, it is equivalent to optimize over integral matchings in $\X$ or fractional matchings in $\bar \X$.

Second, note that $z > 0$. We cannot have $z < 0$ because $\boldsymbol{y}\ge 0$. If we have $z = 0$, then $\boldsymbol{y} = 0$, which can never be optimal because $\b \ge 0$ and $\b \neq 0$. By construction, the matching satisfies $\x \in \bar \X$ with $\x = \boldsymbol{y} / z$. Chapter 4.3.2 of \cite{boyd2004convex} shows the equivalence between the linear fractional program and the linear programming formulation, which concludes the proof.
\end{proof}

\input{cascade}

%% file: cascade.tex
\section{Cascade Model}\label{sec:cascade}

In this section, we consider the Cascade model. In particular, we focus on the design of efficient approximate algorithms for the WDP. In light of the reduction from Section \ref{subsec:approx-mech}, it suffices to show that the induced allocation function is monotonic.

For ease of notation, we take the values $\v$ as the ``bids'' of the WDP and hence the optimization target becomes the social welfare $\wel(\x, \sigma) = \sum_{i} v_i \cdot \pi_i(\x, \sigma)$. Without loss of generality, we assume $v_1 \geq v_2 \geq \cdots \geq v_n$.






The rest of the section is organized as follows. In Section \ref{subsec:perm}, we simplify the objective by fixing the optimal permutation. In Section \ref{sec:4apx}, we introduce an intermediate problem $\welr$ which 4-approximates our original WDP problem, and it is simpler in the sense that it avoids the cascading definition. In Section \ref{sec:ptas}, we give a PTAS for $\welr$ (Theorem \ref{thm:ptas_welr}) and this implies a constant approximation algorithm for WDP (Corollary \ref{cor:4welr}). Unfortunately, we cannot prove that it gives a monotonic allocation. Finally, in Section \ref{sec:polylog-approx}, we give an approximation algorithm to WDP with a monotonic allocation (Theorem \ref{thm:log-approx}), with a weaker approximation ratio $O(\ln m)$. As a corollary, 
\begin{corollary}\label{cor:cascade-mech}
  Under the Cascade model, we have an $O(\ln m)$-approximate, computationally efficient, $\epsilon$-IC, and IR auction for welfare maximization or revenue maximization, respectively.
\end{corollary}
We leave it as an interesting open problem to close the gap between the approximation ratio of the WDP and the approximation ratio of the WDP with monotonic allocation.
\subsection{Optimal Permutation}\label{subsec:perm}

We first pin down the choice of the optimal permutation, which turns out to be agnostic to the implementation details of the allocation function. Therefore we can focus on the allocation design and simplify the notations.

The first thing to notice is that, when a slot $j$ is not allocated (that is, $\sum_{i\in[n]}x_{ij} = 0$), shuffling the slot up and down in the permutation does not change the welfare. Without loss of generality, we only consider permutations which put unallocated slots in the end, and call this set of permutations $\Sigma^*$. Next, for these permutations, we show a lemma about the permutation maximizing welfare.

\begin{lemma}\label{lem:cascade-perm}
For a fixed $\x \in \X$, any permutation $\sigma^* \in \Sigma^*$ maximizing $\wel(\x, \sigma)$ sorts slots in decreasing order of their matched values, i.e. for any $j$ and $j'$ with $\sum_{i\in[n]} v_i x_{ij} > \sum_{i \in [n]} v_i x_{ij'} > 0$, $\sigma^*(j) < \sigma^*(j')$.
\end{lemma}
\begin{proof}
We prove this lemma by contradiction. Suppose there exist $j$ and $j'$ such that $\sigma^*(j) > \sigma^*(j')$ and $\sum_{i\in[n]} v_i x_{ij} >  \sum_{i \in [n]} v_i x_{ij'} > 0$. It is easy to see that we can find $j_1$ and $j_2$ with $\sigma^*(j') \leq \sigma^*(j_1) = \sigma^*(j_2) -1< \sigma^*(j)$, and $\sum_{i\in[n]} v_i x_{ij_1} < \sum_{i \in [n]} v_i x_{ij_2}$.

Now consider a different permutation $\sigma'$ which is the same as $\sigma^*$ except $\sigma'(j_1) =\sigma^*(j_2)$ and $\sigma'(j_2) = \sigma^*(j_1)$. Note that $\x \in \X$ is an integral solution, and let $i_1$ be the advertiser matched with $j_1$ and $i_2$ be the advertiser matched with $j_2$. 

We have $v_{i_1} < v_{i_2}$, and we now compare the welfare of two permutations:
\begin{align*}
&\wel(\x, \sigma^*) - \wel(\x, \sigma') \\
&= v_{i_1} \cdot (\pi_{i_1}(\x, \sigma^*) - \pi_{i_1}(\x, \sigma')) + v_{i_2} \cdot (\pi_{i_2}(\x, \sigma^*) - \pi_{i_2}(\x, \sigma')) \\
&=  \prod_{j': \sigma^*(j')<\sigma^*(j_1)} \left( 1 - \sum_{i' \in [n]}  x_{i'j'} p_{i'j'} \right) \left( v_{i_1} p_{i_1j_1} (1 - (1-p_{i_2j_2})) - v_{i_2} p_{i_2j_2}(1 - (1-p_{i_1j_1}))\right)\\
&= \prod_{j': \sigma^*(j')<\sigma^*(j_1)} \left( 1 - \sum_{i' \in [n]}  x_{i'j'} p_{i'j'} \right) p_{i_1j_1} p_{i_2j_2} (v_{i_1} - v_{i_2})\\
&<0.
\end{align*}
We get $\wel(\x, \sigma^*) < \wel(\x, \sigma')$ and this contradicts with $\sigma^*$ maximizing $\wel$.
\end{proof}

In the rest of the section, we will only consider the welfare-maximizing permutation which sorts its matched value in decreasing order (without loss of generality, we break ties of equal matched values by advertiser indices). And we omit $\sigma$ in notation $\pi$ and the corresponding $\wel$. We have
\[
    \pi_i(\x) = \sum_{j \in [m]} x_{ij} \cdot p_{ij} \cdot  \prod_{i' < i} \left(1- \sum_{j'\in[m]}x_{i'j'}p_{i'j'} \right)\,,
\]
and this implies
\[
    \pi_i(\x) = \sum_{j \in [m]} x_{ij} \cdot p_{ij} \cdot  \left(1-\sum_{i' < i} \pi_{i'}(\x)\right)\,.
\]


\subsection{4-approximating welfare via the restricted welfare}
\label{sec:4apx}

In this section, we introduce a restricted welfare problem $\welr$ and show it gives a 4-approximation to the WDP in the Cascade model. 

The restricted welfare is defined as follows. We define an auxiliary clickthrough rate $\pir$ that does not cascade and only counts welfare for top slots up to total click-through rate 1, in contrast to the click-through rate $\pi$. More formally,
\[
    \pir_i(\x) = \sum_{j \in [m]} x_{ij} \cdot\min \left( p_{ij},  1 - \sum_{i' < i} \pir_{i'}(\x) \right)\,.
\]

And we define the corresponding $\welr(\x)$ as $\wel(\x)$ with $\pi$ replaced by $\pir$. We first show a simple lemma to compare between $\pi(\x)$ and $\pir(\x)$:

\begin{lemma}
\label{lem:atmostone}
For any $n' = 0,...,n$,
\[
\sum_{i=1}^{n'} \pi_i(x) \leq \sum_{i=1}^{n'} \pir_i(\x) \leq 1,
\]
and $\pir_{i}(\x) \geq 0$ for $i \in [n]$.
\end{lemma}

\begin{proof}
We prove by induction on $n'$. The base case $n' = 0$ is trivial. 

Assume the claim is true for $n'=k$. Now consider the case $n'= k+1$. First of all, since $\sum_{i=1}^k \pir_i(\x) \leq 1$, by the definition of $\pir_{k+1}(\x)$, we know $\pir_{k+1}(\x) \geq 0$.

We also have
\begin{align*}
\sum_{i=1}^{k+1} \pir_i(\x) &= \sum_{i=1}^k \pir_i(\x) + \pir_{k+1}(\x)
\leq \sum_{i=1}^k \pir_i(\x) + \left(1 - \sum_{i=1}^k \pir_i(\x)\right) = 1\,.
\end{align*}

For $\sum_{i=1}^{k+1} \pir_i(\x)$  versus $\sum_{i=1}^{k+1} \pi_i(\x)$, there are two cases:
\begin{itemize}
    \item Case 1: $\sum_{i=1}^{k+1} \pir_i(\x) = 1$. In this case, we just need to show $\sum_{i=1}^{k+1} \pi_i(\x) \leq 1$. Since $1-\sum_{i=1}^k \pi_i(\x) \geq 0$, by definition, we have $\pi_{k+1}(\x) \leq 1 \cdot (1-\sum_{i=1}^k \pi_i(\x))$. And this implies $\sum_{i=1}^{k+1} \pi_i(\x) \leq 1$.
    \item Case 2: $\sum_{i=1}^{k+1} \pir_i(\x) < 1$. In this case, we know $\pir_{k+1}(\x) = \sum_{j\in [m]} x_{ij} \cdot p_{ij} \geq \pi_{k+1}(\x)$. By induction, we also have $\sum_{i=1}^k \pi_i(\x) \leq \sum_{i=1}^k \pir_i(\x)$. Put them together, we get $\sum_{i=1}^{k+1} \pi_i(\x) \leq \sum_{i=1}^{k+1} \pir_i(\x)$.
\end{itemize}
\end{proof}

We now prove in Lemma \ref{lem:4ap_lb} and Lemma \ref{lem:4ap_ub} that $\welr(\x)$ is a 4-approximation of $\wel(\x)$.

\begin{lemma}
For any $x\in \X$,
\label{lem:4ap_ub}
\[
\welr(\x) \geq \wel(\x).
\]
\end{lemma}

\begin{proof}
Let $s$ be the last item with cumulative non-cascade click-through rates at most 1, i.e. we have $\sum_{i=1}^s \pir_i(\x) \leq 1$
and, either $s = n$ or $\sum_{i=1}^{s+1} \pir_i(\x) > 1$.
And it is easy to see that by definition for $i \in [s]$, $\pir_i(\x) \geq \pi_i(\x)$. For notation convenience, define $v_{n+1} = 0$ (it's for the case when we need to use $v_{s+1}$ and $s=n$). 

\begin{align*}
\wel(\x) &= \sum_{i\in[n]}v_i \pi_i(\x)
= \sum_{i=1}^s v_i \pi_i(\x)  + \sum_{i=s+1}^n v_i \pi_i(\x) \\
&\leq \sum_{i=1}^s v_i \pir_i(\x) - \sum_{i=1}^s v_{s+1}(\pir_i(\x) - \pi_i(\x)) +   \sum_{i=s+1}^n v_{s+1} \pi_i(\x) \\
&\leq \sum_{i=1}^s v_i \pir_i(\x) + v_{s+1}  \cdot \left( \sum_{i \in [n]} \pi_i(\x) - \sum_{i=1}^s \pir_i(\x)\right) \\
&\leq \sum_{i=1}^s v_i \pir_i(\x) + v_{s+1}  \cdot \left( \sum_{i\in[n]} \pir_i(\x) - \sum_{i=1}^s \pir_i(\x)\right) \\
&\leq \welr(\x).
\end{align*}
\end{proof}

\begin{lemma}
For any $\x \in \X$,
\label{lem:4ap_lb}
\[
\welr(\x) / 4 \leq \wel(\x).
\]
\end{lemma}

\begin{proof}
Let $s$ be the last item with cumulative non-cascade click-through rates at most 1/2, i.e. we have
\[
\sum_{i=1}^s \pir_i(\x) \leq 1/2
\]
and, either $s = n$ or
\[
\sum_{i=1}^{s+1} \pir_i(\x) > 1/2.
\]
And it is easy to see that by definition for $i \in [\min(n, s+1)]$, $\pir_i(\x) \leq 2 \cdot \pi_i(\x)$. For notation convenience, define $v_i = 0$ for $i > n$.

If $s=n$, we simply have
\[
\sum_{i=1}^{s+1} v_i \pir_i(\x)  =\welr(\x) \geq \welr(\x)/2.
\]

If $s < n$, we have
\[
\sum_{i=1}^{s+1} v_i \pir_i(\x) \geq v_{s+1}  \cdot \frac{1}{2} \geq \sum_{i=s+2}^n v_i \pir_i(\x) = \welr(\x) - \sum_{i=1}^{s+1} v_i \pir_i(\x).
\]
And we get $\sum_{i=1}^{s+1} v_i \pir_i(\x)  \geq \welr(\x)/2$. 

Then we have,
\begin{align*}
\wel(\x)  &= \sum_{i\in[n]}v_i \pi_i(\x)
\geq \sum_{i=1}^{s+1} v_i \pi_i(\x)
\geq \sum_{i=1}^{s+1} v_i \pir_i(\x) /2 
\geq \welr(\x) / 4.
\end{align*}
\end{proof}

With Lemma \ref{lem:4ap_lb} and Lemma \ref{lem:4ap_ub}, we know that if we optimize $\welr(\x)$, we get a 4-approximation to $\wel(\x)$. In the next section, we show an algorithm to (approximately) optimize $\welr(\x)$.

\subsection{A PTAS for $\welr$}
\label{sec:ptas}
We show a PTAS for $\welr$ in this section. We are going to use an algorithm from \cite{10.5555/1788814.1788839} for the budgeted matching problem which is similar to the restricted welfare problem. The budgeted matching problem can be defined as the following using our paper's notation for our use case: we first define $\pib_i(\x) =  \sum_{j \in [m]} x_{ij}  p_{ij}$, and we define the corresponding $\welb(\x)$ as $\wel(\x)$ with $\pi$ replaced by $\pib$. The goal of the budgeted matching problem is to find an integer solution $\x$ maximizing $\welb(\x)$ subject to $\sum_{i \in [n]} \pib_i(\x) \leq 1$.

\begin{theorem}[Theorem 1 of \cite{10.5555/1788814.1788839}]
\label{thm:alg-bm}
There is a PTAS for the budgeted matching problem.
\end{theorem}








Here we give a PTAS algorithm for $\welr$. The main idea of the algorithm is to guess the discounted advertiser (as later defined in Definition \ref{def:disct}) and how much it is discounted, and then apply the PTAS of the budget matching problem. 

\begin{algorithm}[H]
\caption{PTAS for $\welr$}
\label{alg:ptas_welr}
\DontPrintSemicolon
\KwIn{Approximation parameter $\varepsilon \in (0,1)$, number of advertisers $n$, number of positions $m$, values $v_i$ for $i \in [n]$, base click-through rates $p_{ij}$ for $i \in [n], j \in [m]$}
\KwOut{Allocation $\x$}
Initialize $\x$ to be an arbitrary allocation\;
\For{ $k = 1,...,n$}{
    \For{$\alpha = \varepsilon/2, \varepsilon, ..., \lfloor 2/\varepsilon \rfloor \cdot \varepsilon/2, 1$}{
        Define $\pka_{ij} = p_{ij}$ for $i\in [n]\backslash \{k\}, j\in[m]$, and $\pka_{ij} = p_{ij} \cdot \alpha$ for $i = k, j\in[m]$\;
        Solve the budget matching with $\pka$ as base click-through rates ($\welbka$)  using the PTAS from \cite{10.5555/1788814.1788839} to get a $(1-\varepsilon/2)$-approximation $\xka$ \;
        \If{$\welr(\xka ) \geq \welr(\x)$}{
            $\x \leftarrow \xka$\;
        }
    }
}
Output $\x$\;
\end{algorithm}

Before we analyze Algorithm \ref{alg:ptas_welr}, we first prove several utility notations and lemmas.

\begin{definition}
\label{def:disct}
For any $\x \in \X$, call an advertiser $k$ a discounted advertiser if $\pir_k(\x) > 0$ and $\pir_k(\x) < \pib_k(\x) = \sum_{j \in [m]} x_{kj}  p_{kj}$.
\end{definition}

\begin{lemma}
\label{lem:disct}
For any $\x \in \X$, there exists at most one discounted advertiser.
\end{lemma}

\begin{proof}
We prove by contradiction. Suppose both $k$ and $k'$ satisfy the condition for being a discounted advertiser. Wlog let $k < k'$. Since $\pir_k(\x) < \sum_{j \in [m]} x_{kj}  p_{kj}$, by definition, we know $\pir_k(\x) = 1 - \sum_{i=1}^{k-1} \pir_i(\x)$. So $\sum_{i=1}^k \pir_i(\x) = 1$. Now we have $\pir_{k'} \leq 1 - \sum_{i=1}^{k'-1} \pir_i(\x) \leq 1 -\sum_{i=1}^k \pir_i(\x) = 0$, and this gives a contradiction.
\end{proof}

\begin{definition}
\label{def:zeroes}    
For any $\x \in \X$, define $\zs(\x)$ to be the procedure of forcing advertiser $i$'s allocation to be 0 if $\pir_i(\x) = 0$, i.e.
\begin{itemize}
\item $\zs(\x)_{ij} = \x_{ij}$ if $\pir_i(\x) > 0$,
\item and $\zs(\x)_{ij} = 0$ if $\pir_i(\x) = 0$.
\end{itemize}
\end{definition}

\begin{lemma}
\label{lem:zeroes}
For any $\x \in \X$, $\x' = \zs(\x)$ has the following two nice properties:
\begin{itemize}
\item $\pir_i(\x) = \pir_i(\x')$ for any $i\in[n]$, and $\welr(\x) = \welr(\x')$
\item For any advertiser $i$ that is not a discounted advertiser in $\x$, $\pir_i(\x') = \pib_i(\x')$.
\end{itemize}
\end{lemma}

\begin{proof}
First of all, $\pir_i(\x) = \pir_i(\x')$ is straightforward to check by the definition using an induction on $i$. Then we have
\[
\welr(\x) = \sum_{i=1}^n v_i \pir_i(x) = \sum_{i=1}^n v_i \pir_i(\x') = \welr(\x').
\]

For the second part of the lemma's claim, let us consider any advertiser $i$ that is not a discounted advertiser in $\x$. By definition, we know either $\pir_i(\x) =0 $ or $0 < \pir_i(\x) = \pib_i(\x)$.

In the first case where $\pir_i(\x) =0$, we know $\x'_{ij} = 0$ for $j\in [m]$, and therefore $\pib_i(\x') = 0 = \pir_i(\x) = \pir_i(\x')$.

In the second case where $0 < \pir_i(\x) = \pib_i(\x)$, we know $\x'_{ij} = \x_{ij}$ for $j \in [m]$, therefore $\pib_i(\x') = \pib_i(\x) = \pir_i(\x) = \pir_i(\x')$.
\end{proof} 

\begin{lemma}
\label{lem:bka}
For any $\x \in \X$ satisfying $\sum_{i=1}^n \pibka_i(\x) \leq 1$, we have
\[
\welbka(\x) \leq \welr(\x).
\]
\end{lemma}

\begin{proof}
For any $\ell \in [n]$, by definition, we have
\begin{gather*}
\pir_{\ell}(\x) = \sum_{j \in [m]} x_{\ell j} \cdot\min \bigg( p_{\ell j},  1 - \sum_{i < \ell} \pir_{i}(\x) \bigg) \\
= \min\bigg(\sum_{j \in [m]} x_{\ell j} \cdot p_{\ell j},  1 - \sum_{i < \ell} \pir_{i}(\x)\bigg) 
= \min\bigg(\pib_{\ell}(\x),  1 - \sum_{i < \ell} \pir_{i}(\x)\bigg)
\end{gather*}

Now we prove by induction that for any $\ell \in [n]$, $\sum_{i=1}^{\ell} \pir_i(\x) \geq \sum_{i=1}^{\ell} \pibka_i(\x)$. The base case with $\ell = 0$ is trivial. Suppose the claim is true for $\ell-1$, for the claim with $\ell$, there are two cases: 
\begin{itemize}
\item Case 1: $\sum_{i=1}^{\ell} \pir_i(\x) =1$. In this case, we simply have $$\sum_{i=1}^{\ell} \pir_i(\x) =1 \geq \sum_{i=1}^n \pibka_i(\x) \geq \sum_{i=1}^{\ell} \pibka_i(\x).$$
\item Case 2: $\sum_{i=1}^{\ell} \pir_i(\x) \neq 1$. In this case, we know $\pir_{\ell}(\x) \neq 1-\sum_{i < \ell} \pir_i(\x)$. Together with $\pir_{\ell}(\x) = \min\left(\pib_{\ell}(\x),  1 - \sum_{i < \ell} \pir_{i}(\x)\right) $, we know $\pir_{\ell}(\x) = \pib_{\ell}(\x)$. By induction hypothesis, we also have $\sum_{i=1}^{\ell-1} \pir_i(\x) \geq \sum_{i=1}^{\ell-1} \pibka_i(\x)$.  Put them together, we get $\sum_{i=1}^{\ell} \pir_i(\x) \geq \sum_{i=1}^{\ell} \pibka_i(\x)$.
\end{itemize}
Finally (here we set $v_{n+1} = 0$ for notation convenience),
\begin{align*}
&\welbka(\x) 
= \sum_{i=1}^n v_i \cdot \pibka(\x) 
= \sum_{\ell=1}^n (v_{\ell} - v_{\ell+1}) \cdot \sum_{i=1}^{\ell} \pibka(\x) \\
\leq &\sum_{\ell=1}^n (v_{\ell} - v_{\ell+1}) \cdot \sum_{i=1}^{\ell} \pir(\x)
= \sum_{i=1}^n v_i \cdot \pir(\x)
= \welr(\x).
\end{align*}

\end{proof}

Now we are ready to prove the main theorem for Algorithm \ref{alg:ptas_welr}.

\begin{theorem}
\label{thm:ptas_welr}
Algorithm \ref{alg:ptas_welr} is a PTAS for $\welr$, i.e., 
\begin{itemize}
    \item Algorithm \ref{alg:ptas_welr} runs in time polynomial in $n$ and $m$ for a fixed $\varepsilon$,
    \item and the output of Algorithm \ref{alg:ptas_welr} has $\welr$ at least $(1-\varepsilon)$ times the optimal $\welr$.
\end{itemize}
\end{theorem}

\begin{proof}
The running time guarantee is straightforward to check. The algorithm has $O(n / \varepsilon)$ iterations, and inside each iteration, both the budget matching algorithm and $\welr$ computation can be done in time polynomial in $n$ and $m$.

For the approximation guarantee, we denote the optimal allocation for $\welr$ as $\x^*$. By Lemma \ref{lem:disct}, $\x^*$ has at most one discounted advertiser. If $\x^*$ does have one discounted advertiser, let $k$ to be that one, and set $\alpha = \left\lfloor \frac{2\pir_k(\x^*)}{ \pib_k(\x^*) \varepsilon} \right\rfloor \cdot \varepsilon / 2$. If $\x^*$ does not have a discounted advertiser, let $k$ be any arbitrary advertiser and set $\alpha = 1$.

Set $\x' = \zs(\x^*)$ using $\zs$ defined in Definition \ref{def:zeroes}. By Lemma \ref{lem:zeroes}, we know $\pir_i(\x') = \pir_i(\x^*)$ for $i\in [n]$, $\welr(\x') = \welr(\x^*)$, and $\pir_i(\x') = \pib_i(\x')$ for $i \in [n] \backslash \{k\}$. For $\pib_k(\x')$, there are two cases depending on whether $k$ is a discounted advertiser in $\x^*$:
\begin{itemize}
\item If $k$ is not a discounted advertiser in $\x^*$, we know we set $\alpha=1$. By Lemma \ref{lem:zeroes}, we have $\pir_k(\x')=\pib_k(\x') = \alpha \cdot \pib_k(\x')$.
\item If $k$ is a discounted advertiser in $\x^*$, we know $\pir_k(\x^*) > 0$ and by the definition of $\zs$, $\x^*_{kj} = \x'_{kj}$ for $j \in [m]$. Therefore, $\pib_k(\x^*) = \pib_k(\x')$. We get
\[
 \frac{2\pir_k(\x^*)}{ \pib_k(\x^*) \varepsilon}  \cdot \frac{\varepsilon}{2} \cdot \pib_k(\x') = \frac{\pir_k(\x')}{\pib_k(\x')} \cdot \pib_k(\x') = \pir_k(\x').
\]
Notice that
\[
\alpha  \leq  \frac{2\pir_k(\x^*)}{ \pib_k(\x^*) \varepsilon}  \cdot \frac{\varepsilon}{2}  \leq \alpha+\varepsilon/2
\]
We get $\alpha \cdot \pib_k(\x') \leq \pir_k(\x') \leq (\alpha+\varepsilon/2) \cdot \pib_k(\x')$
\end{itemize}
And we can conclude in both cases we have $\alpha \cdot \pib_k(\x') \leq \pir_k(\x') \leq (\alpha+\varepsilon/2) \cdot \pib_k(\x')$.

We then compare $\welr(\x')$ with $\welbka(\x')$:
\begin{align*}
&\welbka(\x')\\ &= \sum_{i\in[n] \backslash \{k\}} v_i \cdot \pibka_i(\x') + v_k \cdot \pibka_k(\x') \\&= \sum_{i\in[n] \backslash \{k\}} v_i \cdot \pib_i(\x') + v_k \cdot \pibka_k(\x') \\&= \sum_{i\in[n] \backslash \{k\}} v_i \cdot \pir_i(\x') + v_k \cdot \pibka_k(\x') \\
&= \sum_{i\in[n] \backslash \{k\}} v_i \cdot \pir_i(\x') + v_k \cdot \alpha \cdot \pib_k(\x') \\
&\geq  \sum_{i\in[n] \backslash \{k\}} v_i \cdot \pir_i(\x') + v_k \cdot \pir_k(\x') - (\varepsilon /2) \cdot v_k \cdot \pib_k(\x') \\
&\geq  \sum_{i\in[n] \backslash \{k\}} v_i \cdot \pir_i(\x') + v_k \cdot \pir_k(\x') - (\varepsilon /2) \cdot \welr(\x') \\
&= (1-\varepsilon/2) \cdot \welr(\x').
\end{align*}
The inequality in the second last step comes from the fact that when we only allocate to advertiser $k$ according to $\x'_k$, the clickthrough rates don't need to be capped by 1 and the restricted welfare should be $v_k \sum_{j\in[m]}\x'_{kj}p_{kj} = \v_k \cdot \pib_k(\x')$ which is at most the optimal restricted welfare $\welr(\x^*) = \welr(\x')$.

Now we show $\x'$ is a valid solution for $\welbka$, i.e. $\x'$ satisfies constraint \\$\sum_{i=1}^n \pibka_i(\x') \leq 1$. We have
\begin{align*}
&\sum_{i=1}^n \pibka_i(\x') = \sum_{i\in[n] \backslash \{k\}} \pibka_i(\x') + \pibka_k(\x') \\
= &\sum_{i\in[n] \backslash \{k\}} \pib_i(\x') + \pibka_k(\x')
= \sum_{i\in[n] \backslash \{k\}} \pir_i(\x') + \pibka_k(\x') \\
= &\sum_{i\in[n] \backslash \{k\}} \pir_i(\x') + \alpha \cdot \pib_k(\x')
\leq \sum_{i\in[n] \backslash \{k\}} \pir_i(\x') + \pir_k(\x') \\
\leq& 1 \quad\text{(by applying Lemma \ref{lem:atmostone})}
\end{align*}

Since $\x'$ is a valid solution for $\welbka$, by the approximation guarantee of the PTAS for the budgeted matching problem, we know $\welbka(\x') \leq \welbka(\xka) / (1-\varepsilon / 2)$.

Finally, by applying Lemma \ref{lem:bka}, we have $\welbka(\xka) \leq \welr(\xka)$.

To sum up, we have
\begin{align*}
& \welr(\x^*) = \welr(\x')
\leq \welbka(\x') / (1-\varepsilon/2)\\
\leq & \welbka(\xka) / (1-\varepsilon/2)^2
\leq \welr(\xka) / (1-\varepsilon/2)^2
\leq \welr(\xka) / (1-\varepsilon).
\end{align*}

Note that the output of Algorithm \ref{alg:ptas_welr} has $\welr$ at least $\welr(\xka)$. This completes the proof of the approximation guarantee of Algorithm \ref{alg:ptas_welr}.
\end{proof}

Combine Theorem \ref{thm:ptas_welr} and what we have in Section \ref{sec:4apx} (Lemma \ref{lem:4ap_ub} and \ref{lem:4ap_lb}), we get the following corollary.

\begin{corollary}
\label{cor:4welr}
Algorithm \ref{alg:ptas_welr} gives a $\frac{4}{1-\varepsilon}$-approximation to the optimal $\wel$ in the Cascade model.
\end{corollary}

\input{polylog_approx}

%% file: polylog_approx.tex
\subsection{Monotonic $O(\ln m)$-approximation of the WDP}
\label{sec:polylog-approx}
In this section, we give a randomized $O(\ln m)$-approximate algorithm for the WDP, which induces a monotonic allocation function. According to Section \ref{subsec:approx-mech}, the corresponding welfare-maximization auction and revenue-maximization auction are implied with the same approximation guarantee.
\begin{restatable}{theorem}{logApprox}
\label{thm:log-approx}
    There is a polynomial-time randomized algorithm that finds an allocation with expected approximation ratio of $O(\ln m)$ to the optimal welfare in the Cascade model. Moreover, the allocation algorithm is monotonic.
\end{restatable}
Similar to our approach in Section~\ref{sec:ptas}, we work with restricted welfare and the budgeted matching problem using the base click-through rate $p_{ij}$'s instead of the actual (cascade) click-through rate. A naive way to address the budget constraint is to only pick a single edge which must be feasible since $p_{ij}\leq 1$. The edge with maximum $v_i\cdot p_{ij}$ is a trivial $m$-approximation. Note that in the previous argument, we essentially turn the budget constraint into a cardinality constraint (of $1$ edge), and we can build upon this idea to get the stronger $O(\ln m)$-approximation by bucketizing the edges by their $p_{ij}$'s. 

In particular, we partition the advertiser-position pairs (i.e. edges) into $G_{1},\ldots, G_{\log_2 (4m)}$ (where for notation simplicity we assume $m$ is a power of $2$). The subgraph $G_{\ell}$ contains all edges with $\frac{1}{2^{\ell}}< p_{ij}\leq \frac{1}{2^{\ell-1}}$ for $\ell\in [1,\log_2 (2m)]$, and the last subgraph $G_{\log_2 (4m)}$ contains all remaining edges (i.e. those with $p_{ij}\leq 1/(2m)$). 

If we can find the allocation with the optimal welfare in each subgraph respectively, the sum of these allocations' welfare is at least the welfare of the optimal allocation in the original graph (i.e. with all edges), and this gives us the following.
\begin{restatable}{lemma}{randomPick}
\label{lem:random-pick}
If we can find a $\beta$-approximately optimal welfare allocation in each subgraph respectively, then the algorithm of picking a bucket $\ell\in [1,\log_2 (4m)]$ uniformly at random and using the $\beta$-approximately optimal allocation of $G_{\ell}$ would have an expected $\beta\cdot\log_2 (4m) $ approximation ratio to the optimal welfare in the full graph. Furthermore, if the algorithm to find the allocation in each subgraph is monotonic, the overall algorithm (i.e. with the random picking) is also monotonic.
\end{restatable}
This reduces our task to finding a monotonic allocation algorithm that approximately optimize the welfare in a single subgraph. To make our allocation algorithm monotonic, we need to explicitly specify the permutation $\sigma$ of the matched positions, which may not be the same as the optimal permutation that ranks the positions in decreasing order of their matched advertisers' value. 

We use the (canonical) greedy algorithm to find a {\em maximal} matching subject to cardinality constraint in the subgraph $G_{\ell}$, and use the same order of edges being added to the matching as the permutation of the matched positions.

\begin{algorithm}[H]
\caption{Greedy maximal matching in $G_{\ell}$}
\label{alg:greedy}
\DontPrintSemicolon
\KwIn{Index $\ell\in[1,\log_2 (4m)]$, advertiser values $v_i$ for $i \in [n]$, and $p_{ij}$'s \textbf{for edges $(i,j)$  in $G_{\ell}$}.}
\KwOut{Matching $\x_{\ell}$ and permutation $\sigma_{\ell}$ of matched positions.}
Start with empty matching $\x_{\ell}$ and number of matched edges $k=0$\;
Go through the edges in decreasing order of weights $w_{ij}=v_i\cdot p_{ij}$ (with any deterministic tie-breaking)\;
\While{There are still edges to consider, denote $(i,j)$ as the next edge}{
    \If{$(i,j)$ can be added to $\x_{\ell}$ (i.e. both advertiser $i$ and position $j$ are unmatched)}{
    $k\leftarrow k+1$, $\sigma_{\ell}(j)\leftarrow k$. Add the edge $(i,j)$ to $\x_{\ell}$\;
    \If{$k=\min(2^{\ell},m)$}{Return $\x_{\ell}$ and $\sigma_{\ell}$\;}
    }
    }
Return $\x_{\ell}$ and $\sigma_{\ell}$\;
\end{algorithm}
\begin{restatable}{lemma}{greedyAlg}
\label{lem:greedy-alg}
    For any subgraph $G_{\ell}$, Algorithm~\ref{alg:greedy} (in polynomial time) finds an allocation of welfare at least $\frac{1}{28}$ of the optimal allocation in $G_{\ell}$, and the allocation algorithm is monotonic.
\end{restatable}
It is straightforward to see that Theorem~\ref{thm:log-approx} follows from Lemma~\ref{lem:random-pick} and Lemma~\ref{lem:greedy-alg}.

%% file: appendix.tex
\input{polylog_appendix}

%% file: polylog_appendix.tex
\section{Missing proofs of Section~\ref{sec:polylog-approx}}
In this section we give the missing proofs in the analysis of the monotonic $O(\ln m)$-approximation algorithm for welfare maximization. 
\randomPick*
\begin{proof}
We consider the allocation $\x^*$ with the optimal welfare in the full graph, and $\x^*_\ell$ as the intersection of $\x^*$ and the subgraph $G_{\ell}$, i.e. $\x^*_\ell$ contains the subset of the advertiser-position pairs of $\x^*$ that exist in $G_{\ell}$ and the ranking of these pairs in $\x^*_\ell$ remains the same as their (relative) ranking in the full allocation $\x^*$. Denote $\wel_{\ell}(\x^*_\ell)$ as the welfare of $\x^*_\ell$ in $G_{\ell}$, since each $G_{\ell}$ is a subgraph of the full graph, any shown ad in $\x^*_\ell$ must have a higher click-through rate in $\x^*_\ell$ compared to its click-through rate in $\x^*$ in the Cascade model (since earlier ads may not exist in subgraph). Moreover, as all the $\x^*_\ell$'s together cover all shown ads allocated in $\x^*$, we must have
\[
\wel(\x^*)\leq \sum_{\ell} \wel_{\ell}(\x^*_{\ell}).
\]
If we have a $\beta$-approximately optimal allocation $\x_{\ell}$ in each subgraph $G_{\ell}$, we know $\wel_{\ell}(\x_{\ell})\geq \ \wel_{\ell}(\x^*_{\ell}/\beta)$, and it's also clear that $\wel(\x_{\ell})= \wel_{\ell}(\x_{\ell})$. Thus,
\[
\sum_{\ell\in[1,\log_2(4m)]} \wel(\x_{\ell}) \geq \wel(\x^*)/\beta .
\]
Pick a $\x_{\ell}$ uniformly at random would then have welfare on expectation at least $\frac{1}{\beta\cdot\log_2 (4m)}\cdot \wel(\x^*)$.

Since all $\x_{\ell}$'s are results of a monotonic algorithm, randomly sampling from them is also monotonic.
\end{proof}
As noted, we work with the base click-through rate when we find a matching in our allocation, and recall the corresponding definitions:
\[
\welbase(\x) = \sum_{i\in [n]} v_i\pibase_i(\x), \text{ 
 where }
\pibase_i(\x) = \sum_{j\in[m]} x_{ij}\cdot p_{ij}.
\]
As noted, (for the monotonicity of the allocation algorithm) we may use a permutation $\sigma$ of the matched positions that can be different from the optimal permutation which ranks the positions in decreasing order of their matched advertisers' value. To keep things clear, in our analysis we explicitly write any allocation as a pair $(\x,\sigma)$ where $\x$ specifies the matching between advertisers and positions, and $\sigma$ gives the ranking of the matched positions. We also use $\sigma_0$ to denote the (generic) optimal permutation which (when paired with any matching $\x$,) ranks the matched positions by their matched advertiser's value. Using these notations, we can restate Lemma~\ref{lem:4ap_ub} as
\[
\welr(\x,\sigma_0) \geq \wel(\x,\sigma_0) \quad \text{for any matching }\x
\]
Furthermore, since $\pibase_i(\x)$ is without any truncation or cascading effect, it's straightforward to see that for any matching and permutation $(\x,\sigma)$
\begin{equation}
\label{eq:welfares}
\welbase(\x,\sigma) = \welbase(\x,\sigma_0)\geq \welr(\x,\sigma_0) \geq \wel(\x,\sigma_0). 
\end{equation}

Similar to the budgeted matching problem (as a proxy to optimize the restricted welfare), we consider a cardinality constraint where in subgraph $G_\ell$ we optimize over the space of matchings in $G_\ell$ with at most $2^{\ell}$ edges (i.e. $\sum_{i,j}x_{ij}\leq 2^{\ell}$). Maximizing restricted welfare over matchings in $G_\ell$ remain (effectively) the same with or without the cardinality constraint, and together with Lemma~\ref{lem:4ap_ub} we have
\begin{lemma}
\label{lem:baseToWel}
    Let $\X_\ell$ be the space of feasible matchings in $G_{\ell}$, and denote $|\x|$ of a matching $\x$ as the number of matched edges. For any matching $\x$ and permutation $\sigma$ we have
    \[
    \max_{\x\in \X_\ell,|\x|\leq 2^{\ell}}\welbase(\x,\sigma) \geq \max_{\x\in \X_\ell}\wel(\x,\sigma_0).
    \]
Note $\sigma_0$ is the (generic) permutation that orders the positions by their matched advertisers' values, and is the optimal permutation (w.r.t any matching) when maximizing $\wel$.
\end{lemma}
\begin{proof}
Observe that for any (integral) matching $\x\in \X_\ell$, there can be at most $\min(m,2^{\ell})$ edges with non-zero contribution to the $\pir_i(\x)$'s, since for $\ell\leq \log_2(2m)$ every edge $(i,j)$ in $G_\ell$ has $p_{ij}>\frac{1}{2^{\ell}}$. This gives
\[
\max_{\x\in \X_\ell,|\x|\leq 2^{\ell}}\welr(\x,\sigma_0) = \max_{\x\in \X_\ell}\welr(\x,\sigma_0)
\]
Together with~\eqref{eq:welfares}, we have
\begin{align*}
    \max_{\x\in \X_\ell,|\x|\leq 2^{\ell}}\welbase(\x,\sigma)  \geq & \max_{\x\in \X_\ell,|\x|\leq 2^{\ell}}\welr(\x,\sigma_0)\\
    = & \max_{\x\in \X_\ell}\welr(\x,\sigma_0) \\
    \geq & \max_{\x\in \X_\ell}\wel(\x,\sigma_0)
\end{align*}
\end{proof}
We use the (canonical) greedy strategy (Algorithm~\ref{alg:greedy}) to find a matching $\x_\ell$ (with $|\x_\ell|\leq 2^{\ell}$) and an associated permutation $\sigma_\ell$ in $G_\ell$. We need to show that $\x_\ell$ gives a good approximation, and this largely follows the classic result that a greedy maximal matching is a $2$-approximation of the maximum weight matching, with slight adaptation to accommodate the cardinality constraint.
\begin{lemma}
\label{lem:baseApprox}
    Let $(\x_\ell,\sigma_\ell)$ be the result of Algorithm~\ref{alg:greedy} on $G_\ell$, we have
    \[
    \welbase(\x_\ell)\geq \frac{1}{2}\cdot \max_{\x\in \X_\ell,|\x|\leq 2^{\ell}}\welbase(\x).
    \]
    Note here we omit the permutation since $\welbase$ is independent of it.
\end{lemma}
\begin{proof}
For each matched edge $(i,j)$ in $\x_\ell$, we put two tokens each of weight $v_i\cdot p_{ij}$ on the edge $(i,j)$. It's clear the total weight of all tokens we create is $2\cdot\welbase(\x_\ell)$, and we will show the tokens is enough to pay for the optimal $\welbase$. 

Let $\x$ be the optimal matching, and we consider each edge $(i,j)$ in $\x$. We can pay for the contribution of $(i,j)$ to $\welbase(\x)$ using a token we created. If $(i,j)$ is also in $\x_\ell$, we use one of the tokens on $(i,j)$, and each edge in $\x_\ell \cap \x$ pays at most one token in this case. If $(i,j)$ is not in $\x_\ell$, it must be one of the following cases when the greedy algorithm reaches $(i,j)$:
\begin{itemize}
    \item $j$ is already matched to $i'$ in $\x_\ell$; We know $(i',j)$ must have larger weight than $(i,j)$, and we can pay with a token on $(i',j)$. Each edge in $\x_\ell\setminus \x$ pays at most one token from this case.
    \item $i$ is already matched to $j'$ in $\x_\ell$; Similar to the above case, we can pay with a token on $(i,j')$. Again, each edge in $\x_\ell\setminus \x$ pays at most one token from this case.
    \item $\x_\ell$ already has $\min(m,2^\ell)$ edges; In this case any of the $2\cdot \min(m,2^\ell)$ tokens is larger than the weight of $(i,j)$. Note $\x$ also has at most $\min(m,2^\ell)$ edges, and each edge in $\x$ consumes (exactly) one of the tokens we created, thus after paying for all edges in $\x$ from the previous $2$ cases, we must have enough tokens left to pay for all the edges in this case.
\end{itemize}
From the above, we know $2\cdot \welbase(\x_\ell)\geq \welbase(\x)$, which proves our lemma.
\end{proof}
We proceed with the following result to bound the gap between $\wel(\x_\ell,\sigma_\ell)$ and $\welbase(\x_\ell,\sigma_\ell)$, which is analogous to Lemma~\ref{lem:4ap_lb}.
\begin{lemma}
\label{lem:welToBase}
    Let $(\x_\ell,\sigma_\ell)$ be the result of Algorithm~\ref{alg:greedy} on $G_\ell$, we have
    \[
    \wel(\x_\ell,\sigma_\ell)\geq \frac{1}{14}\cdot 
    \welbase(\x_\ell)
    \]
\end{lemma}
\begin{proof}
We largely follows the proof of Lemma~\ref{lem:4ap_lb}. We (re-)index the advertisers in the same order as indicated by $\sigma_\ell$, and let $s$ be the last advertiser with cumulative non-cascade click-through rates at most $0.5$, i.e. we have
\[
\sum_{i=1}^s \pibase_i(\x_\ell) \leq 0.5
\]
and, either $s = n$ or
\[
\sum_{i=1}^{s+1} \pibase_i(\x_\ell) > 0.5.
\]
And it is easy to see that by definition of $\pi$, for any advertiser $i$ up to $s+1$ we have $\pi_i(\x_\ell,\sigma_\ell)\geq 0.4\cdot \pibase_i(\x_\ell)$. 
If $s=n$, we simply have 
\[
\wel(\x_\ell,\sigma_\ell) = \sum_{i=1}^{s} v_i \pi_i(\x_\ell,\sigma_\ell)  \geq \sum_{i=1}^{s} v_i\cdot 0.4\cdot \pibase_i(\x_\ell) = 0.4\cdot \welbase(\x_\ell).
\]
Note if $\ell=\log_2 (4m)$ we must have the above case, since all $p_{ij}$'s in that subgraph is at most $1/(2m)$ and we can match at most $m$ edges. 

On the other hand, if $s < n$, we know the minimum value $\min_{i\in[s+1]}v_i$ of the first $s+1$ matched advertisers is at least half of the maximum value of matched advertisers after $s+1$. This is because the greedy algorithm goes through the edges in decreasing order of $v_i\cdot p_{ij}$ and all $p_{ij}$'s in $G_{\ell}$ for $\ell\neq \log_2(4m)$ are within a factor of $2$ of each other. Thus, for any two matched advertisers $i,i'$ with $i$ being matched before $i'$, we must have $v_i\geq v_{i'}/2$. Furthermore, we know $\sum_{i=1}^n \pibase_i(\x_\ell)\leq 2$ because of the cardinality constraint of $2^\ell$ (and $p_{ij}\leq 2/2^\ell$ in $G_\ell$), so $\sum_{i=s+2}^n \pibase_i(\x_\ell)\leq 2-0.5=1.5$. We then have
\begin{align*}
\welbase(\x_\ell) = &\sum_{i=1}^{s+1} v_i \pibase_i(\x_\ell) + \sum_{i=s+2}^{n} v_i \pibase_i(\x_\ell) \\
\leq & \sum_{i=1}^{s+1} v_i \pibase_i(\x_\ell) +  1.5\cdot2\cdot \min_{i\in[s+1]}v_{i}\\
\leq & \sum_{i=1}^{s+1} v_i \pibase_i(\x_\ell) +  \frac{3}{0.5}\cdot \sum_{i=1}^{s+1} v_i \pibase_i(\x_\ell)\\
= & 7 \cdot \sum_{i=1}^{s+1} v_i \pibase_i(\x_\ell)
\end{align*}
So we get $\sum_{i=1}^{s+1} v_i \pibase_i(\x_\ell)  \geq \frac{1}{7}\cdot \welbase(\x_\ell)$. Then we have,
\[
\wel(\x_\ell,\sigma_\ell) = \sum_{i\in[n]}v_i \pi_i(\x_\ell,\sigma_\ell)\geq \sum_{i=1}^{s+1}v_i \pi_i(\x_\ell,\sigma_\ell) \geq 0.5\cdot \sum_{i=1}^{s+1}v_i \pibase_i(\x_\ell)
\geq \frac{1}{14}\cdot \welbase(\x_\ell).
\]
\end{proof}
We are now ready to prove the main result of the greedy algorithm.
\greedyAlg*
\begin{proof}
Algorithm~\ref{alg:greedy} clearly runs in polynomial time since we only need to sort the (up to) $m\cdot n$ edges in the subgraph and the remaining work is linear in the number of edges. The approximation guarantee follows straightforwardly from Lemma~\ref{lem:baseToWel},~\ref{lem:baseApprox} and~\ref{lem:welToBase}
\[
\wel(\x_\ell,\sigma_\ell)\geq \frac{1}{14}\cdot\welbase(\x_\ell)\geq \frac{1}{28}\cdot\max_{\x\in \X_\ell,|\x|\leq 2^{\ell}}\welbase(\x)\geq \frac{1}{28}\cdot \max_{\x\in \X_\ell}\wel(\x,\sigma_0).
\]
As to the monotonicity, since we greedily go through the edges by their weight $v_i\cdot p_{ij}$, if the value $v_i$ of any (matched) advertiser $i$ increases, $i$ must be matched at least as early as before, so the cascading effect (i.e. the discounting) to $\pi_i$ cannot become stronger. Moreover, all edges adjacent to $i$ keep their relative order when $v_i$ increases, so the matched edge for $i$ must have a base click-through rate at least as large as before. These two together indicate that $\pi_i(\x_\ell)$ is monotonic in $v_i$.
\end{proof}

%% file: main-arxiv.bbl

\begin{thebibliography}{40}


\ifx \showCODEN    \undefined \def \showCODEN     #1{\unskip}     \fi
\ifx \showDOI      \undefined \def \showDOI       #1{#1}\fi
\ifx \showISBNx    \undefined \def \showISBNx     #1{\unskip}     \fi
\ifx \showISBNxiii \undefined \def \showISBNxiii  #1{\unskip}     \fi
\ifx \showISSN     \undefined \def \showISSN      #1{\unskip}     \fi
\ifx \showLCCN     \undefined \def \showLCCN      #1{\unskip}     \fi
\ifx \shownote     \undefined \def \shownote      #1{#1}          \fi
\ifx \showarticletitle \undefined \def \showarticletitle #1{#1}   \fi
\ifx \showURL      \undefined \def \showURL       {\relax}        \fi
\providecommand\bibfield[2]{#2}
\providecommand\bibinfo[2]{#2}
\providecommand\natexlab[1]{#1}
\providecommand\showeprint[2][]{arXiv:#2}

\bibitem[Abeliuk et~al\mbox{.}(2014)]%
        {abeliuk2014optimizing}
\bibfield{author}{\bibinfo{person}{Andres Abeliuk}, \bibinfo{person}{Gerardo
  Berbeglia}, \bibinfo{person}{Manuel Cebrian}, {and} \bibinfo{person}{Pascal
  Van~Hentenryck}.} \bibinfo{year}{2014}\natexlab{}.
\newblock \showarticletitle{Optimizing Expected Utility in a Multinomial Logit
  Model with Position Bias and Social Influence}.
\newblock \bibinfo{journal}{\emph{arXiv preprint arXiv:1411.0279}}
  (\bibinfo{year}{2014}).
\newblock


\bibitem[Aggarwal et~al\mbox{.}(2008)]%
        {AggarwalFMP08}
\bibfield{author}{\bibinfo{person}{Gagan Aggarwal}, \bibinfo{person}{Jon
  Feldman}, \bibinfo{person}{S. Muthukrishnan}, {and} \bibinfo{person}{Martin
  P{\'{a}}l}.} \bibinfo{year}{2008}\natexlab{}.
\newblock \showarticletitle{Sponsored Search Auctions with Markovian Users}. In
  \bibinfo{booktitle}{\emph{Internet and Network Economics, 4th International
  Workshop, {WINE} 2008, Shanghai, China, December 17-20, 2008. Proceedings}}
  \emph{(\bibinfo{series}{Lecture Notes in Computer Science},
  Vol.~\bibinfo{volume}{5385})}, \bibfield{editor}{\bibinfo{person}{Christos~H.
  Papadimitriou} {and} \bibinfo{person}{Shuzhong Zhang}} (Eds.).
  \bibinfo{publisher}{Springer}, \bibinfo{pages}{621--628}.
\newblock
\urldef\tempurl%
\url{https://doi.org/10.1007/978-3-540-92185-1\_68}
\showDOI{\tempurl}


\bibitem[Athey and Ellison(2011)]%
        {athey2011position}
\bibfield{author}{\bibinfo{person}{Susan Athey} {and} \bibinfo{person}{Glenn
  Ellison}.} \bibinfo{year}{2011}\natexlab{}.
\newblock \showarticletitle{Position auctions with consumer search}.
\newblock \bibinfo{journal}{\emph{The Quarterly Journal of Economics}}
  \bibinfo{volume}{126}, \bibinfo{number}{3} (\bibinfo{year}{2011}),
  \bibinfo{pages}{1213--1270}.
\newblock


\bibitem[Bakhtin et~al\mbox{.}(2022)]%
        {meta2022human}
\bibfield{author}{\bibinfo{person}{Anton Bakhtin}, \bibinfo{person}{Noam
  Brown}, \bibinfo{person}{Emily Dinan}, \bibinfo{person}{Gabriele Farina},
  \bibinfo{person}{Colin Flaherty}, \bibinfo{person}{Daniel Fried},
  \bibinfo{person}{Andrew Goff}, \bibinfo{person}{Jonathan Gray},
  \bibinfo{person}{Hengyuan Hu}, {et~al\mbox{.}}}
  \bibinfo{year}{2022}\natexlab{}.
\newblock \showarticletitle{Human-level play in the game of Diplomacy by
  combining language models with strategic reasoning}.
\newblock \bibinfo{journal}{\emph{Science}} \bibinfo{volume}{378},
  \bibinfo{number}{6624} (\bibinfo{year}{2022}), \bibinfo{pages}{1067--1074}.
\newblock


\bibitem[Berger et~al\mbox{.}(2008)]%
        {10.5555/1788814.1788839}
\bibfield{author}{\bibinfo{person}{Andr\'{e} Berger}, \bibinfo{person}{Vincenzo
  Bonifaci}, \bibinfo{person}{Fabrizio Grandoni}, {and} \bibinfo{person}{Guido
  Sch\"{a}fer}.} \bibinfo{year}{2008}\natexlab{}.
\newblock \showarticletitle{Budgeted matching and budgeted matroid intersection
  via the gasoline puzzle}. In \bibinfo{booktitle}{\emph{Proceedings of the
  13th International Conference on Integer Programming and Combinatorial
  Optimization}} (Bertinoro, Italy) \emph{(\bibinfo{series}{IPCO'08})}.
  \bibinfo{publisher}{Springer-Verlag}, \bibinfo{address}{Berlin, Heidelberg},
  \bibinfo{pages}{273–287}.
\newblock
\showISBNx{3540688862}


\bibitem[Blumrosen et~al\mbox{.}(2008)]%
        {blumrosen2008position}
\bibfield{author}{\bibinfo{person}{Liad Blumrosen}, \bibinfo{person}{Jason
  Hartline}, {and} \bibinfo{person}{Shuzhen Nong}.}
  \bibinfo{year}{2008}\natexlab{}.
\newblock \showarticletitle{Position auctions and non-uniform conversion
  rates}. In \bibinfo{booktitle}{\emph{ACM EC Workshop on Advertisement
  Auctions}}.
\newblock


\bibitem[Boyd et~al\mbox{.}(2004)]%
        {boyd2004convex}
\bibfield{author}{\bibinfo{person}{Stephen Boyd}, \bibinfo{person}{Stephen~P
  Boyd}, {and} \bibinfo{person}{Lieven Vandenberghe}.}
  \bibinfo{year}{2004}\natexlab{}.
\newblock \bibinfo{booktitle}{\emph{Convex optimization}}.
\newblock \bibinfo{publisher}{Cambridge university press}.
\newblock


\bibitem[Chen et~al\mbox{.}(2023)]%
        {chen2023put}
\bibfield{author}{\bibinfo{person}{Jiangjie Chen}, \bibinfo{person}{Siyu Yuan},
  \bibinfo{person}{Rong Ye}, \bibinfo{person}{Bodhisattwa~Prasad Majumder},
  {and} \bibinfo{person}{Kyle Richardson}.} \bibinfo{year}{2023}\natexlab{}.
\newblock \showarticletitle{{Put your money where your mouth is: Evaluating
  strategic planning and execution of LLM agents in an auction arena}}.
\newblock \bibinfo{journal}{\emph{arXiv preprint arXiv:2310.05746}}
  (\bibinfo{year}{2023}).
\newblock


\bibitem[Chu et~al\mbox{.}(2020)]%
        {ChuNZ20}
\bibfield{author}{\bibinfo{person}{Leon~Yang Chu}, \bibinfo{person}{Hamid
  Nazerzadeh}, {and} \bibinfo{person}{Heng Zhang}.}
  \bibinfo{year}{2020}\natexlab{}.
\newblock \showarticletitle{Position Ranking and Auctions for Online
  Marketplaces}.
\newblock \bibinfo{journal}{\emph{Manag. Sci.}} \bibinfo{volume}{66},
  \bibinfo{number}{8} (\bibinfo{year}{2020}), \bibinfo{pages}{3617--3634}.
\newblock
\urldef\tempurl%
\url{https://doi.org/10.1287/MNSC.2019.3372}
\showDOI{\tempurl}


\bibitem[Clarke(1971)]%
        {clarke1971multipart}
\bibfield{author}{\bibinfo{person}{Edward~H Clarke}.}
  \bibinfo{year}{1971}\natexlab{}.
\newblock \showarticletitle{Multipart pricing of public goods}.
\newblock \bibinfo{journal}{\emph{Public choice}} \bibinfo{volume}{11},
  \bibinfo{number}{1} (\bibinfo{year}{1971}), \bibinfo{pages}{17--33}.
\newblock


\bibitem[Craswell et~al\mbox{.}(2008)]%
        {craswell2008experimental}
\bibfield{author}{\bibinfo{person}{Nick Craswell}, \bibinfo{person}{Onno
  Zoeter}, \bibinfo{person}{Michael Taylor}, {and} \bibinfo{person}{Bill
  Ramsey}.} \bibinfo{year}{2008}\natexlab{}.
\newblock \showarticletitle{An experimental comparison of click position-bias
  models}. In \bibinfo{booktitle}{\emph{Proceedings of the 2008 international
  conference on web search and data mining}}. \bibinfo{pages}{87--94}.
\newblock


\bibitem[Deng et~al\mbox{.}(2024)]%
        {deng2024llms}
\bibfield{author}{\bibinfo{person}{Yuan Deng}, \bibinfo{person}{Vahab
  Mirrokni}, \bibinfo{person}{Renato~Paes Leme}, \bibinfo{person}{Hanrui
  Zhang}, {and} \bibinfo{person}{Song Zuo}.} \bibinfo{year}{2024}\natexlab{}.
\newblock \showarticletitle{Llms at the bargaining table}. In
  \bibinfo{booktitle}{\emph{Agentic Markets Workshop at ICML}},
  Vol.~\bibinfo{volume}{2024}.
\newblock


\bibitem[Dubey et~al\mbox{.}(2024)]%
        {dubey2024auctions}
\bibfield{author}{\bibinfo{person}{Avinava Dubey}, \bibinfo{person}{Zhe Feng},
  \bibinfo{person}{Rahul Kidambi}, \bibinfo{person}{Aranyak Mehta}, {and}
  \bibinfo{person}{Di Wang}.} \bibinfo{year}{2024}\natexlab{}.
\newblock \showarticletitle{Auctions with llm summaries}. In
  \bibinfo{booktitle}{\emph{Proceedings of the 30th ACM SIGKDD Conference on
  Knowledge Discovery and Data Mining}}. \bibinfo{pages}{713--722}.
\newblock


\bibitem[Duetting et~al\mbox{.}(2024)]%
        {duetting2024mechanism}
\bibfield{author}{\bibinfo{person}{Paul Duetting}, \bibinfo{person}{Vahab
  Mirrokni}, \bibinfo{person}{Renato Paes~Leme}, \bibinfo{person}{Haifeng Xu},
  {and} \bibinfo{person}{Song Zuo}.} \bibinfo{year}{2024}\natexlab{}.
\newblock \showarticletitle{Mechanism design for large language models}. In
  \bibinfo{booktitle}{\emph{Proceedings of the ACM on Web Conference 2024}}.
  \bibinfo{pages}{144--155}.
\newblock


\bibitem[Edelman et~al\mbox{.}(2007)]%
        {edelman2007internet}
\bibfield{author}{\bibinfo{person}{Benjamin Edelman}, \bibinfo{person}{Michael
  Ostrovsky}, {and} \bibinfo{person}{Michael Schwarz}.}
  \bibinfo{year}{2007}\natexlab{}.
\newblock \showarticletitle{Internet advertising and the generalized
  second-price auction: Selling billions of dollars worth of keywords}.
\newblock \bibinfo{journal}{\emph{American Economic Review}}
  \bibinfo{volume}{97}, \bibinfo{number}{1} (\bibinfo{year}{2007}),
  \bibinfo{pages}{242--259}.
\newblock


\bibitem[Fan et~al\mbox{.}(2024)]%
        {fan2024can}
\bibfield{author}{\bibinfo{person}{Caoyun Fan}, \bibinfo{person}{Jindou Chen},
  \bibinfo{person}{Yaohui Jin}, {and} \bibinfo{person}{Hao He}.}
  \bibinfo{year}{2024}\natexlab{}.
\newblock \showarticletitle{Can large language models serve as rational players
  in game theory? a systematic analysis}. In
  \bibinfo{booktitle}{\emph{Proceedings of the AAAI Conference on Artificial
  Intelligence}}, Vol.~\bibinfo{volume}{38}. \bibinfo{pages}{17960--17967}.
\newblock


\bibitem[Gemp et~al\mbox{.}(2024)]%
        {gemp2024states}
\bibfield{author}{\bibinfo{person}{Ian Gemp}, \bibinfo{person}{Yoram Bachrach},
  \bibinfo{person}{Marc Lanctot}, \bibinfo{person}{Roma Patel},
  \bibinfo{person}{Vibhavari Dasagi}, \bibinfo{person}{Luke Marris},
  \bibinfo{person}{Georgios Piliouras}, {and} \bibinfo{person}{Karl Tuyls}.}
  \bibinfo{year}{2024}\natexlab{}.
\newblock \showarticletitle{States as Strings as Strategies: Steering Language
  Models with Game-Theoretic Solvers}.
\newblock \bibinfo{journal}{\emph{arXiv preprint arXiv:2402.01704}}
  (\bibinfo{year}{2024}).
\newblock


\bibitem[Ghosh and Mahdian(2008)]%
        {GhoshM08}
\bibfield{author}{\bibinfo{person}{Arpita Ghosh} {and}
  \bibinfo{person}{Mohammad Mahdian}.} \bibinfo{year}{2008}\natexlab{}.
\newblock \showarticletitle{Externalities in online advertising}. In
  \bibinfo{booktitle}{\emph{Proceedings of the 17th International Conference on
  World Wide Web, {WWW} 2008, Beijing, China, April 21-25, 2008}},
  \bibfield{editor}{\bibinfo{person}{Jinpeng Huai}, \bibinfo{person}{Robin
  Chen}, \bibinfo{person}{Hsiao{-}Wuen Hon}, \bibinfo{person}{Yunhao Liu},
  \bibinfo{person}{Wei{-}Ying Ma}, \bibinfo{person}{Andrew Tomkins}, {and}
  \bibinfo{person}{Xiaodong Zhang}} (Eds.). \bibinfo{publisher}{{ACM}},
  \bibinfo{pages}{161--168}.
\newblock
\urldef\tempurl%
\url{https://doi.org/10.1145/1367497.1367520}
\showDOI{\tempurl}


\bibitem[Gomes et~al\mbox{.}(2009)]%
        {GomesIM09}
\bibfield{author}{\bibinfo{person}{Renato Gomes}, \bibinfo{person}{Nicole
  Immorlica}, {and} \bibinfo{person}{Evangelos Markakis}.}
  \bibinfo{year}{2009}\natexlab{}.
\newblock \showarticletitle{Externalities in Keyword Auctions: An Empirical and
  Theoretical Assessment}. In \bibinfo{booktitle}{\emph{Internet and Network
  Economics, 5th International Workshop, {WINE} 2009, Rome, Italy, December
  14-18, 2009. Proceedings}} \emph{(\bibinfo{series}{Lecture Notes in Computer
  Science}, Vol.~\bibinfo{volume}{5929})},
  \bibfield{editor}{\bibinfo{person}{Stefano Leonardi}} (Ed.).
  \bibinfo{publisher}{Springer}, \bibinfo{pages}{172--183}.
\newblock
\urldef\tempurl%
\url{https://doi.org/10.1007/978-3-642-10841-9\_17}
\showDOI{\tempurl}


\bibitem[Gravin et~al\mbox{.}(2024)]%
        {gravin2024bidder}
\bibfield{author}{\bibinfo{person}{Nikolai Gravin},
  \bibinfo{person}{Yixuan~Even Xu}, {and} \bibinfo{person}{Renfei Zhou}.}
  \bibinfo{year}{2024}\natexlab{}.
\newblock \showarticletitle{Bidder Selection Problem in Position Auctions: A
  Fast and Simple Algorithm via Poisson Approximation}. In
  \bibinfo{booktitle}{\emph{Proceedings of the ACM on Web Conference 2024}}.
  \bibinfo{pages}{89--98}.
\newblock


\bibitem[Hajiaghayi et~al\mbox{.}(2024)]%
        {hajiaghayi2024ad}
\bibfield{author}{\bibinfo{person}{MohammadTaghi Hajiaghayi},
  \bibinfo{person}{S{\'e}bastien Lahaie}, \bibinfo{person}{Keivan Rezaei},
  {and} \bibinfo{person}{Suho Shin}.} \bibinfo{year}{2024}\natexlab{}.
\newblock \showarticletitle{Ad Auctions for LLMs via Retrieval Augmented
  Generation}.
\newblock \bibinfo{journal}{\emph{arXiv preprint arXiv:2406.09459}}
  (\bibinfo{year}{2024}).
\newblock


\bibitem[Jeziorski and Segal(2015)]%
        {JeziorskiSegal15}
\bibfield{author}{\bibinfo{person}{Przemyslaw Jeziorski} {and}
  \bibinfo{person}{Ilya Segal}.} \bibinfo{year}{2015}\natexlab{}.
\newblock \showarticletitle{What Makes Them Click: Empirical Analysis of
  Consumer Demand for Search Advertising}.
\newblock \bibinfo{journal}{\emph{American Economic Journal: Microeconomics}}
  \bibinfo{volume}{7}, \bibinfo{number}{3} (\bibinfo{date}{August}
  \bibinfo{year}{2015}), \bibinfo{pages}{24–53}.
\newblock
\urldef\tempurl%
\url{https://doi.org/10.1257/mic.20100119}
\showDOI{\tempurl}


\bibitem[Kempe and Mahdian(2008)]%
        {KempeM08}
\bibfield{author}{\bibinfo{person}{David Kempe} {and} \bibinfo{person}{Mohammad
  Mahdian}.} \bibinfo{year}{2008}\natexlab{}.
\newblock \showarticletitle{A Cascade Model for Externalities in Sponsored
  Search}. In \bibinfo{booktitle}{\emph{Internet and Network Economics, 4th
  International Workshop, {WINE} 2008, Shanghai, China, December 17-20, 2008.
  Proceedings}} \emph{(\bibinfo{series}{Lecture Notes in Computer Science},
  Vol.~\bibinfo{volume}{5385})}, \bibfield{editor}{\bibinfo{person}{Christos~H.
  Papadimitriou} {and} \bibinfo{person}{Shuzhong Zhang}} (Eds.).
  \bibinfo{publisher}{Springer}, \bibinfo{pages}{585--596}.
\newblock
\urldef\tempurl%
\url{https://doi.org/10.1007/978-3-540-92185-1\_65}
\showDOI{\tempurl}


\bibitem[Lor{\`e} and Heydari(2023)]%
        {lore2023strategic}
\bibfield{author}{\bibinfo{person}{Nunzio Lor{\`e}} {and}
  \bibinfo{person}{Babak Heydari}.} \bibinfo{year}{2023}\natexlab{}.
\newblock \showarticletitle{Strategic behavior of large language models: Game
  structure vs. contextual framing}.
\newblock \bibinfo{journal}{\emph{arXiv preprint arXiv:2309.05898}}
  (\bibinfo{year}{2023}).
\newblock


\bibitem[Lu et~al\mbox{.}(2024)]%
        {lu2024eliciting}
\bibfield{author}{\bibinfo{person}{Yuxuan Lu}, \bibinfo{person}{Shengwei Xu},
  \bibinfo{person}{Yichi Zhang}, \bibinfo{person}{Yuqing Kong}, {and}
  \bibinfo{person}{Grant Schoenebeck}.} \bibinfo{year}{2024}\natexlab{}.
\newblock \showarticletitle{Eliciting informative text evaluations with large
  language models}. In \bibinfo{booktitle}{\emph{Proceedings of the 25th ACM
  Conference on Economics and Computation}}. \bibinfo{pages}{582--612}.
\newblock


\bibitem[Luce(1959)]%
        {luce1959individual}
\bibfield{author}{\bibinfo{person}{R~Duncan Luce}.}
  \bibinfo{year}{1959}\natexlab{}.
\newblock \bibinfo{booktitle}{\emph{Individual choice behavior}}.
  Vol.~\bibinfo{volume}{4}.
\newblock \bibinfo{publisher}{Wiley New York}.
\newblock


\bibitem[Mao et~al\mbox{.}(2023)]%
        {mao2023alympics}
\bibfield{author}{\bibinfo{person}{Shaoguang Mao}, \bibinfo{person}{Yuzhe Cai},
  \bibinfo{person}{Yan Xia}, \bibinfo{person}{Wenshan Wu}, \bibinfo{person}{Xun
  Wang}, \bibinfo{person}{Fengyi Wang}, \bibinfo{person}{Tao Ge}, {and}
  \bibinfo{person}{Furu Wei}.} \bibinfo{year}{2023}\natexlab{}.
\newblock \showarticletitle{Alympics: Language agents meet game theory}.
\newblock \bibinfo{journal}{\emph{arXiv preprint arXiv:2311.03220}}
  (\bibinfo{year}{2023}).
\newblock


\bibitem[McFadden(1974)]%
        {mcfadden_conditional_1974}
\bibfield{author}{\bibinfo{person}{Daniel McFadden}.}
  \bibinfo{year}{1974}\natexlab{}.
\newblock \showarticletitle{Conditional logit analysis of qualitative choice
  behavior}.
\newblock In \bibinfo{booktitle}{\emph{Fontiers in {Econometrics}}},
  \bibfield{editor}{\bibinfo{person}{Paul Zarembka}} (Ed.).
  \bibinfo{publisher}{Academic press}, \bibinfo{address}{New York},
  \bibinfo{pages}{105--142}.
\newblock


\bibitem[Milgrom and Segal(2002)]%
        {milgrom2002envelope}
\bibfield{author}{\bibinfo{person}{Paul Milgrom} {and} \bibinfo{person}{Ilya
  Segal}.} \bibinfo{year}{2002}\natexlab{}.
\newblock \showarticletitle{Envelope theorems for arbitrary choice sets}.
\newblock \bibinfo{journal}{\emph{Econometrica}} \bibinfo{volume}{70},
  \bibinfo{number}{2} (\bibinfo{year}{2002}), \bibinfo{pages}{583--601}.
\newblock


\bibitem[Myerson(1981)]%
        {myerson1981optimal}
\bibfield{author}{\bibinfo{person}{Roger~B Myerson}.}
  \bibinfo{year}{1981}\natexlab{}.
\newblock \showarticletitle{Optimal auction design}.
\newblock \bibinfo{journal}{\emph{Mathematics of operations research}}
  \bibinfo{volume}{6}, \bibinfo{number}{1} (\bibinfo{year}{1981}),
  \bibinfo{pages}{58--73}.
\newblock


\bibitem[Plackett(1975)]%
        {plackett1975analysis}
\bibfield{author}{\bibinfo{person}{Robin~L Plackett}.}
  \bibinfo{year}{1975}\natexlab{}.
\newblock \showarticletitle{The analysis of permutations}.
\newblock \bibinfo{journal}{\emph{Journal of the Royal Statistical Society
  Series C: Applied Statistics}} \bibinfo{volume}{24}, \bibinfo{number}{2}
  (\bibinfo{year}{1975}), \bibinfo{pages}{193--202}.
\newblock


\bibitem[Raman et~al\mbox{.}(2024)]%
        {raman2024rationality}
\bibfield{author}{\bibinfo{person}{Narun Raman}, \bibinfo{person}{Taylor
  Lundy}, \bibinfo{person}{Samuel Amouyal}, \bibinfo{person}{Yoav Levine},
  \bibinfo{person}{Kevin Leyton-Brown}, {and} \bibinfo{person}{Moshe
  Tennenholtz}.} \bibinfo{year}{2024}\natexlab{}.
\newblock \showarticletitle{Rationality Report Cards: Assessing the Economic
  Rationality of Large Language Models}.
\newblock \bibinfo{journal}{\emph{arXiv preprint arXiv:2402.09552}}
  (\bibinfo{year}{2024}).
\newblock


\bibitem[Schwartz(2023)]%
        {bing2023}
\bibfield{author}{\bibinfo{person}{Eric~Hal Schwartz}.}
  \bibinfo{year}{2023}\natexlab{}.
\newblock \showarticletitle{Microsoft Bing Generative AI Chatbot Experiments
  With Ads}.
\newblock
  \bibinfo{journal}{\emph{https://voicebot.ai/2023/04/03/microsoft-bing-generative-ai-chatbot-experiments-with-ads/}}
  (\bibinfo{year}{2023}).
\newblock


\bibitem[Schwartz(2024)]%
        {perplexity2024}
\bibfield{author}{\bibinfo{person}{Eric~Hal Schwartz}.}
  \bibinfo{year}{2024}\natexlab{}.
\newblock \showarticletitle{Perplexity Will Embed Ads in Generative AI Search
  Engine: Report}.
\newblock
  \bibinfo{journal}{\emph{https://voicebot.ai/2024/04/02/perplexity-will-embed-ads-in-generative-ai-search-engine/}}
  (\bibinfo{year}{2024}).
\newblock


\bibitem[Soumalias et~al\mbox{.}(2024)]%
        {soumalias2024truthful}
\bibfield{author}{\bibinfo{person}{Ermis Soumalias}, \bibinfo{person}{Michael~J
  Curry}, {and} \bibinfo{person}{Sven Seuken}.}
  \bibinfo{year}{2024}\natexlab{}.
\newblock \showarticletitle{Truthful Aggregation of LLMs with an Application to
  Online Advertising}.
\newblock \bibinfo{journal}{\emph{arXiv preprint arXiv:2405.05905}}
  (\bibinfo{year}{2024}).
\newblock


\bibitem[Sun et~al\mbox{.}(2024)]%
        {sun2024mechanism}
\bibfield{author}{\bibinfo{person}{Haoran Sun}, \bibinfo{person}{Yurong Chen},
  \bibinfo{person}{Siwei Wang}, \bibinfo{person}{Wei Chen}, {and}
  \bibinfo{person}{Xiaotie Deng}.} \bibinfo{year}{2024}\natexlab{}.
\newblock \showarticletitle{Mechanism Design for LLM Fine-tuning with Multiple
  Reward Models}.
\newblock \bibinfo{journal}{\emph{arXiv preprint arXiv:2405.16276}}
  (\bibinfo{year}{2024}).
\newblock


\bibitem[Thompson and Leyton-Brown(2009)]%
        {thompson2009computational}
\bibfield{author}{\bibinfo{person}{David Robert~Martin Thompson} {and}
  \bibinfo{person}{Kevin Leyton-Brown}.} \bibinfo{year}{2009}\natexlab{}.
\newblock \showarticletitle{Computational analysis of perfect-information
  position auctions}. In \bibinfo{booktitle}{\emph{Proceedings of the 10th ACM
  conference on Electronic commerce}}. \bibinfo{pages}{51--60}.
\newblock


\bibitem[Varian(2007)]%
        {varian2007position}
\bibfield{author}{\bibinfo{person}{Hal~R Varian}.}
  \bibinfo{year}{2007}\natexlab{}.
\newblock \showarticletitle{Position auctions}.
\newblock \bibinfo{journal}{\emph{International Journal of Industrial
  Organization}} \bibinfo{volume}{25}, \bibinfo{number}{6}
  (\bibinfo{year}{2007}), \bibinfo{pages}{1163--1178}.
\newblock


\bibitem[Vickrey(1961)]%
        {vickrey1961counterspeculation}
\bibfield{author}{\bibinfo{person}{William Vickrey}.}
  \bibinfo{year}{1961}\natexlab{}.
\newblock \showarticletitle{Counterspeculation, auctions, and competitive
  sealed tenders}.
\newblock \bibinfo{journal}{\emph{The Journal of Finance}}
  \bibinfo{volume}{16}, \bibinfo{number}{1} (\bibinfo{year}{1961}),
  \bibinfo{pages}{8--37}.
\newblock


\bibitem[Wu and Hartline(2024)]%
        {wu2024elicitationgpt}
\bibfield{author}{\bibinfo{person}{Yifan Wu} {and} \bibinfo{person}{Jason
  Hartline}.} \bibinfo{year}{2024}\natexlab{}.
\newblock \showarticletitle{ElicitationGPT: Text Elicitation Mechanisms via
  Language Models}.
\newblock \bibinfo{journal}{\emph{arXiv preprint arXiv:2406.09363}}
  (\bibinfo{year}{2024}).
\newblock


\end{thebibliography}
